%% file: ConnectedKMedian_arXiv.tex
\documentclass[a4paper,UKenglish,cleveref, autoref, thm-restate]{lipics-v2021}

\pdfoutput=1 
\hideLIPIcs  
\usepackage[ruled, linesnumbered]{algorithm2e}
\usepackage{thmtools} 
\usepackage{thm-restate}
\usepackage{booktabs}
\usepackage{mathtools}
\DeclareMathOperator*{\argmax}{argmax}
\DeclareMathOperator*{\argmin}{argmin}

\newcommand{\sep}{\text{sep}}
\newcommand{\cost}{\text{cost}}
\newcommand{\mycomment}[1]{}

\title{Connected k-Median with Disjoint and Non-disjoint Clusters} 

\author{Jan Eube}{University of Bonn, Germany}{eube@cs.uni-bonn.de}{}{}

\author{Kelin Luo}{University at Buffalo, United States}{kelinluo@buffalo.edu}{}{}

\author{Dorian Reineccius}{Deutsches Zentrum für Luft- und Raumfahrt, Germany}{dorian.reineccius@dlr.de}{}{}

\author{Heiko R\"oglin}{University of Bonn, Germany}{roeglin@cs.uni-bonn.de}{}{}

\author{Melanie Schmidt}{Heinrich-Heine Universität Düsseldorf, Germany}{mschmidt@hhu.de}{}{}

\authorrunning{J. Eube, K. Luo, D. Reineccius, H. R\"oglin, M. Schmidt} 

\Copyright{Jan Eube, Kelin Luo, Dorian Reineccius, Heiko R\"oglin, Melanie Schmidt} 



\category{} 

\relatedversion{} 

\acknowledgements{}

\nolinenumbers 

\EventEditors{John Q. Open and Joan R. Access}
\EventNoEds{2}
\EventLongTitle{42nd Conference on Very Important Topics (CVIT 2016)}
\EventShortTitle{CVIT 2016}
\EventAcronym{CVIT}
\EventYear{2016}
\EventDate{December 24--27, 2016}
\EventLocation{Little Whinging, United Kingdom}
\EventLogo{}
\SeriesVolume{42}
\ArticleNo{23}

\begin{document}

\title{Connected $k$-Median with Disjoint and Non-disjoint Clusters\footnote{This work has been funded by the Deutsche Forschungsgemeinschaft (DFG, German Research Foundation) – 390685813; 459420781 and by the Lamarr Institute for Machine Learning and Artificial Intelligence lamarr-institute.org.}}

\maketitle

\begin{abstract}
The connected $k$-median problem is a constrained clustering problem that combines distance-based $k$-clustering with connectivity information. The problem allows to input a metric space and an unweighted undirected connectivity graph that is completely unrelated to the metric space. The goal is to compute $k$ centers and corresponding clusters such that each cluster forms a connected subgraph of $G$, and such that the $k$-median cost is minimized.

The problem has applications in very different fields like geodesy (particularly districting), social network analysis (especially community detection), or bioinformatics. We study a version with overlapping clusters where points can be part of multiple clusters which is natural for the use case of community detection. This problem variant is $\Omega(\log n)$-hard to approximate, and our main result is an $\mathcal{O}(k^2 \log n)$-approximation algorithm for the problem. 
We complement it with an $\Omega(n^{1-\epsilon})$-hardness result for the case of disjoint clusters without overlap with general connectivity graphs, as well as an exact algorithm in this setting if the connectivity graph is a tree.
\end{abstract}

\input{introduction}

\section{Non-disjoint Connected $k$-Median Problem}\label{sec:non_disjoint}

\input{nd_assignment}

\subsection{Finding Centers}\label{sec:nd_find_centers_short}

\input{nd_centers}







\bibliography{literature}

\appendix

\input{integral_centers}

\input{inapproxRatio_Non_disjoint}

\section{Disjoint Connected $k$-Median Problem}\label{apx:disjoint}
\subsection{General Case}\label{apx:disjoint_general}

\input{disjoint_general}

\subsection{Solving the Disjoint Connected k-Median Problem on Trees}\label{apx:dyn_prog}


\input{k_median_tree}

\input{alt_lp_flows}

\end{document}

%% file: introduction.tex
\section{Introduction}

Clustering problems, like the $k$-center and the $k$-median problem, are classical optimization problems that have been widely studied both in theory and practice. In a typical center-based clustering problem, the input consists of a given set of data points in some metric space and a positive integer $k$. 
The goal is to partition the data points into $k$ groups (clusters) and to find a center for each of these groups so as to optimize some objective: in the $k$-center problem, the goal is to minimize the maximal distance of any point to the center of its cluster; in the $k$-median problem, the goal is to minimize the sum of the distances of the points to their corresponding cluster centers. 

Both the $k$-center and the $k$-median problem are NP-hard but various constant-factor approximation algorithms are known. For the $k$-center problem, there are 2-approximation algorithms~\cite{Gonzalez85,HochbaumS85}, and this is best possible~\cite{HsuN79}.
The best known approximation algorithm for the $k$-median problem achieves a $(2.675+\epsilon)$-approximation~\cite{byrka2017improved}, while it is known that no approximation factor better than $1+2/e\approx 1.735$ can be achieved, unless P=NP \cite{GuhaK98}. 

Beyond these vanilla clustering problems, there has been a lot of interest in variants of clustering problems motivated by specific applications. These variants often introduce additional constraints. 
For example, in capacitated clustering problems, constraints are added that limit the number of data points that can be assigned to each cluster~\cite{li2016approximating}. 
There are also models to address outliers in data sets~\cite{li2018distributed}, and in the last decade, different variants of fair clustering problems have been studied~\cite{chierichetti2017fair, schmidt2020fair, vakilian2022improved}. 
These are just a few examples of the numerous variants of clustering problems that have been studied to address diverse challenges in different application domains. 
Each of these variants requires specialized algorithms, often inspiring novel methodologies and insights in the broader field of clustering and data analysis.

In this paper, we study a \emph{connectivity} constraint. 
The basic idea of this constraint is that clusters have to correspond to connected subgraphs of a given undirected and unweighted \emph{connectivity graph}. This graph is independent of the metric that defines the distances between the input points. 

This model has been studied in different contexts. One line of research, see Ge et al.~\cite{ge2008joint}, considers the connected $k$-center problem to allow for a joint analysis of two type of data: attribute data and relationship data. Attribute data can be used for clustering by finding a suitable metric distance function on the attribute vectors and then applying a popular clustering method like $k$-means or $k$-center. Relationship data arises in the context of social networks or other network structures and gives information on pairwise relationships. Clustering on relationship data can be done via methods like correlation clustering. Both types of data thus allow for well studied clustering methods, but the focus by Ge et al.~\cite{ge2008joint} is a joint analysis. Connected clustering gives one option to combine relationship and attribute data by defining a constraint based on the relationships and then optimizing over the attribute vector distances.

More precisely, connected clustering expects as input a set of points in a metric space (stemming from the attribute data) and an unweighted graph $G$ (representing the relationship data), and the goal is to cluster with respect to the metric distance function but under the constraint that clusters are connected in $G$. We call $G$ the \emph{connectivity graph} in the following.

Ge et al.~\cite{ge2008joint} study the connected $k$-center problem. For this problem, the goal is to compute $k$ centers $c_1,\ldots,c_k$ and corresponding clusters $C_1,\ldots, C_k$ such that all $C_i$ are connected in the connectivity graph $G$, and such that the largest radius is minimized (the radius of $C_i$ is $\max_{x \in C_i} d(x,c_i)$). They provide an optimal algorithm via dynamic programming for the special case of a tree connectivity graph which was rediscovered by Drexler et al.~\cite{Drexler2024} in a later publication. Ge et al. also present algorithms for the general case with unknown approximation guarantee\footnote{There is a proof in the paper that claims a guarantee of $6$, but~\cite{Drexler2024} gives a counter example.} and provide an experimental evaluation.

The best-known approximation guarantee for the connected $k$-center problem is a $O(\log^2 k)$-approximation given by Drexler et al.~\cite{Drexler2024}. For Euclidean metrics with constant dimension and metrics with constant doubling dimension, they provide an $O(1)$-approximation algorithm. Their paper studies the problem from a very different context: Motivated by an application from geodesy and sea level analysis, they use the distance function to model differences in tide gauge measurements and the connectivity graph to model geographic closeness, and their modelling results in the same problem definition as Ge et al.~\cite{ge2008joint}.

In this paper, we study the connected \emph{$k$-median} problem, i.e., the related problem with the $k$-median objective instead of the $k$-center objective. We give a formal definition now.

\begin{definition}[Connected $k$-Median Problem (disjoint)]\label{def:connectedkmedian}
In an instance of the (disjoint) connected $k$-median problem, we are given a set $V$ of points, a metric $d: V\times V \rightarrow \mathbb{R}_{\ge 0}$ on $V$, a positive integer $k\ge 2$, 
and an unweighted and undirected connectivity graph $G = (V,E)$. 
A feasible solution is a partition of $V$ into $k$ disjoint clusters $\mathcal{V}=\{V_1, \ldots, V_k\}$ with corresponding centers $C=\{c_1, c_2, \ldots, c_k\}\subseteq V$ which satisfies that for every $i \in [k]:=\{1,\ldots, k\}$ the subgraph of $G$ induced by $V_i$ is connected and $c_i \in V_i$.
 The \emph{connected $k$-median problem} aims to find a feasible solution $(C,\mathcal{V})$ that minimizes 
 \[ \Phi_\text{med}(C,\mathcal{V}):= \sum_{i\in [k]} \sum_{v\in V_i} d(c_i, v).\]
\end{definition} 


The connected $k$-median has also appeared in different contexts. For example, Validi et al.~\cite{ValidiBL22} study connected $k$-median for a districting problem. Here the goal is to partition given administrative or geographic units into districts. The clusters should usually correspond to connected areas on the map which can be formulated via introducing centers for each unit and a graph to model which units are neighbors on the map. The $k$-median objective is one way to model compactness of districts. Validi et al.~\cite{ValidiBL22} study various districting methods, and connected $k$-median is one of them. They solve it via an integer linear programming approach, i.e., to optimality but with exponential worst-case running time.
A related variant is the connected $k$-means problem, which for example has been considered in~\cite{LiaoP12} where a heuristic approach to minimize it is employed.

Gupta et al.~\cite{gupta2011clustering} study $k$-median/$k$-means from the point of view of approximation algorithms. They consider the special case where the number of clusters is $k=2$ and the connectivity graph is star-like (there exists a point that is connected to all other points). The paper shows that even this case is $\Omega(\log n)$-hard to approximate for both problems and proposed an $O(\log n)$-approximation algorithm for this special case. 
%
%
There is nothing known about the general case, and our work is motivated by the following question: 
\begin{quote}
    Can the connected $k$-median problem be approximated for general connectivity graphs?
\end{quote}

Unfortunately, we quickly arrive at the conclusion that no reasonable approximation guarantee is possible if $P\neq NP$ (see Thm~\ref{thm:Hardnessk2}). We do derive a dynamic programming based exact algorithm for trees that runs in time $O(n^2k^2)$ (see Thm~\ref{thm:DP}), thus extending the landscape for the problem, but the answer to our main question is still negative.

However, our inapproximability proof depends on a subtlety of the problem definition: The clusters $V_i$ have to be disjoint. This is a natural assumption in clustering, and in clustering without constraints, clusters can be assumed to be disjoint (adding a point to multiple clusters will in general only increase the objective). But here, if we allow for overlapping clusters, then satisfying connectivity can become easier because we can add critical vertices like articulation points to multiple clusters. This motivates the definition of \emph{connected $k$-median with non-disjoint clusters}. 

\begin{definition}[Disjoint vs.\ non-disjoint clusters]
In the \emph{connected $k$-median problem with disjoint clusters}, we require the clusters $V_1,\ldots, V_k$ to be pairwise disjoint. In the \emph{connected $k$-median problem with non-disjoint clusters}, the clusters are allowed to overlap and we only require $V=\bigcup_{i\in[k]} V_i$.
\end{definition} 

One may note that if a vertex appears in multiple clusters it will also contribute to the objective function multiple times. Thus also in the non-disjoint variant, we are incentivized to avoid assigning vertices unnecessarily often. 

The non-disjoint version has also been studied previously. For example, Drexler et al.~\cite{Drexler2024} study it for $k$-center and give a $2$-approximation for this case. From the application side, it depends on the interpretation whether disjoint or non-disjoint clusters are preferred. In districting, units should certainly not be in multiple clusters. But for the application of community detection in social networks the modeling is natural since individuals can belong to different groups. We thus study the following question.

\begin{quote}
    Can the connected $k$-median problem  with non-disjoint clusters be approximated for general connectivity graphs?
\end{quote}

The proof by Gupta et al.~\cite{gupta2011clustering} gives a lower bound of $\Omega(\log n)$ for the disjoint case. For the sake of completeness, we provide an adapted version in Appendix~\ref{sec:appendix_hardness} that also works if the clusters may overlap. Our main result is a $O(k^2 \log n)$-approximation for the connected $k$-median problem with non-disjoint clusters.

Before we summarize our results, we need one more definition. For a clustering problem with constraints, the corresponding \emph{assignment} version of the problem  asks for assigning points to given centers while respecting the constraints and minimizing the cost. In our case, we get the following problem.




\begin{definition}[Assignment version]
In the \emph{assignment version} of the connected $k$-median problem, we assume that $k$ centers, denoted as $C\subseteq V$, are given and the goal is to find a feasible assignment of the points in $V$ to the centers $C$ such that the $k$-median costs are minimized. This version can be defined both for the disjoint and the non-disjoint variant of the connected $k$-median problem.
\end{definition}  

Without a constraint, the assignment version of the $k$-median problem can be solved optimally by assigning every point to its closest center. 
For clustering with constraints, assignment problems can be difficult: For example, the assignment problem for fair $k$-center is NP-hard~\cite{Bercea0KKRS019}, and the connected $k$-center  with disjoint clusters is also NP-hard~\cite{Drexler2024} (both problems are even hard to approximate better than $3$). We will establish an even higher lower bound for the assignment problem in the connected $k$-median setting. 

\subsection{Our results}

\begin{table}[h]
    \centering
    \begin{subtable}{\textwidth}
    \centering
\begin{tabular}{ |m{3cm}|m{4cm}|m{4cm}|  }
 \hline
 &Assignment version& Regular version\\
 \hline
Approximation  & $O(k\log(n))$ (Theorem \ref{thm:AssignmentNonDisjoint})   & $O(k^2\log(n))$ (Theorem \ref{thm:ClusteringNonDisjoint})\\
\hline
 Hardness($k=2$) &\multicolumn{2}{c|}{$\Omega(\log(n))$ (Theorem \ref{thm:non_disjoint_connected_k_median})}\\
 \hline
\end{tabular}
\subcaption{Our results in the non-disjoint setting}
\end{subtable}
\begin{subtable}{\textwidth}
\centering
\begin{tabular}{ |m{3cm}|m{4cm}|m{4cm}|  }
 \hline
 &General Graph& Tree graph\\
 \hline
Approximation  & $O(n \log^2(k))$  (Theorem \ref{thm:disjoint_general_apx}) & $1$ (Theorem \ref{thm:DP})\\
\hline
 Hardness($k=2$) &$\Omega(n^{1-\epsilon})$ (Theorem \ref{thm:Hardnessk2})& -\\
 \hline
\end{tabular}
\subcaption{Our results in the disjoint setting. All results also hold for the assignment version.}
\end{subtable}
\caption{Overview of the results presented in this work.}\label{table:results}
\end{table}
We prove new approximation and hardness results for the connected $k$-median problem, both for the disjoint and the non-disjoint variant. Table \ref{table:results} provides a brief overview of them. First we establish a strong hardness of approximation result for the connected $k$-median problem with disjoint clusters, which even holds in the special case that $k=2$.

\begin{restatable}{theorem}{hardnessdisjoint} 
\label{thm:Hardnessk2}
  For any $\epsilon>0$, there is no polynomial-time $O(n^{1- \epsilon})$-approximation algorithm for the connected $k$-median problem  with disjoint clusters even if $k=2$, unless P = NP. For the assignment version, this lower bound also holds.  
\end{restatable}
   

Another important special case that has also been studied for the connected $k$-center problem is the case that the connectivity graph is a tree. For this special case, we obtain a polynomial-time algorithm based on dynamic programming.

\begin{restatable}{theorem}{DPdisjoint} 
  \label{thm:DP}
  When the connectivity graph $G$ is a tree, then the connected $k$-median problem with disjoint clusters can be solved 
  optimally in time $O(n^2k^2)$. This is true even if the distances are not a metric or the centers are given.
\end{restatable}

Motivated by applications like community detection, in which the clusters do not necessarily have to be disjoint, we also study the connected $k$-median problem with non-disjoint clusters. This problem is $\Omega(\log{n})$-hard to approximate even if $k = 2$ (see Appendix~\ref{sec:appendix_hardness}). Our main result are approximation algorithms for the clustering and the assignment problem whose approximation factors match this lower bound in terms of $n$.

\begin{restatable}{theorem}{AssignmentNonDisjoint} 
  \label{thm:AssignmentNonDisjoint}
    For the assignment version of the non-disjoint connected $k$-median problem, there is a polynomial-time algorithm with approximation factor $O(k \log n)$.
    This is true even if the distances are not a metric.
\end{restatable}

\begin{restatable}{theorem}{ClusteringNonDisjoint} 
  \label{thm:ClusteringNonDisjoint}
    For the non-disjoint connected $k$-median problem, there is a polynomial-time algorithm with approximation factor $O(k^2 \log n)$.
\end{restatable}

\subsection{Outline}
The remainder of this paper is organized as follows: Section \ref{sec:non_disjoint} deals with the non-disjoint variant. In Section \ref{sec:NDCkMAP} we present our algorithm for the assignment version and in Section \ref{sec:nd_find_centers_short} we give an overview on how the centers can be found. The full description and analysis of the latter algorithm can be found in Appendix \ref{apx:finding_centers}. The hardness proof for the disjoint setting is provided in Appendix \ref{sec:appendix_hardness}. 

Our results for the disjoint connected k-median problem are presented in Appendix \ref{apx:disjoint}. Appendix \ref{apx:disjoint_general} focuses on the hardness result on general graphs and Appendix \ref{apx:dyn_prog} presents the exact algorithm for the case that the connectivity graph is a tree. 

\mycomment{
\begin{table}[h]
    \centering
\begin{tabular}{ |c|c|c|c|  }
 \hline
 \multicolumn{2}{|c|}{non disjoint}& \multicolumn{2}{|c|}{disjoint}\\
 \hline
Centers given   & Centers not given   & General graph&   Tree graph\\
\hline
 $[\Omega(\log(n)),O(k\log(n))]$ &$[\Omega(\log(n)),O(k^2\log(n))]$ & $[\Omega(n^{1-\epsilon}), O(n \log^2(k))]$&  $[1,1]$\\
 \hline
\end{tabular}
\end{table}
}

%% file: nd_assignment.tex

\subsection{Assignment Version}\label{sec:NDCkMAP}

The high-level idea to solve the assignment version of the connected $k$-median problem is to first define an LP formulation for it and compute an optimum fractional solution for the LP relaxation in polynomial time. 
Using an observation that the connected clusters are similar to so called node weighted Steiner trees we can use this LP solution to obtain $k$ instances of the node weighted Steiner tree problem such that every node needs to be contained in at least one Steiner tree and there exist a fractional solution for each of these instances such that the total cost of these solutions is bounded by $k$ times the optimum LP solution.
Using a primal-dual algorithm by Klein and Ravi \cite{Klein.Ravi.1995} we are able to find integral Steiner trees that cost at most $O(\log n)$ times as much as the respective fractional solutions and which directly form $k$ connected clusters assigning each node to at least one center. In total, this gives us an approximation factor of $O(k \log n)$.

Now we will define the LP formulation. For that, we use the following definition of cuts:
\begin{definition}[Vertex Cuts]
    For two arbitrary nodes $v, w \in V$, a \emph{$(v,w)$-cut} is a subset $S\subseteq V$ such that all paths from $v$ to $w$ in $G$ contain at least one node in $S$. The set of all $(v,w)$-cuts is denoted by $\mathcal{S}_{v,w}$.
\end{definition}
Using this, we can reformulate the requirement that a node $v$ must be connected to its assigned center $c$ by ensuring that for any $(v,c)$-cut $S \in \mathcal{S}_{v,c}$, at least one node in $S$ is also assigned to $c$.
This gives us the following LP-formulation of the non-disjoint connected $k$-median problem: 

\begin{alignat}{3}\label{LP_nd_assignment}
        \min \quad && \multicolumn{2}{c}{ $\displaystyle\sum\limits_{c \in C, v \in V} d(v,c) x_v^c$ } \nonumber\\
        \textrm{s.t.} \quad &&\sum\limits_{c \in C} x_v^c & \geq 1, && \quad\forall v \in V,\nonumber\\ 
        &&\sum\limits_{v' \in S} x_{v'}^c &\geq x_{v}^c, && \quad\forall v \in V, c \in C, S \in \mathcal{S}_{v,c},\nonumber\\
        &&x_v^c &\in \{0,1\}, && \quad\forall c \in C, v \in V,
\end{alignat}

where for all $v \in V$, $c \in C$ the variable $x_v^c$ indicates whether $v$ is assigned to $c$ ($x_v^c= 1$) or not ($x_v^c = 0$).  While the first constraint ensures, that each point $v$ must belong to at least one cluster, the second constraint ensures that the clusters are connected, i.e., if $v$ is connected to center $c$, then in each $(v,c)$-cut, there must be at least one point assigned to $c$. To relax this problem we drop the integrality constraint and only require that $x_v^c \geq 0$ for all $v \in V$, $c \in C$.

It is worth noting that the linear program may have exponential size, as there can be exponentially many
$(v,c)$-cuts for a given pair $v, c$. However, since the separation problem can be solved in polynomial time and we also provide an alternative flow-based formulation of the LP with polynomial size in Appendix \ref{sec:NDCkMAP-alternative-LP-formulation}, it remains possible to find the optimum fractional LP-solution in polynomial time.

Let $\tilde{x}$ be such an optimum fractional solution. By Markov's inequality it is easy to verify that for any $v \in V$ there exists at least one $c \in C$ such that $\tilde{x}_v^c \geq \frac{1}{k}$. By multiplying all values with $k$, we get a new LP solution $x$ with $x_v^c := k \tilde{x}_v^c$ for all $v \in V,c \in C$ where each node is assigned to at least one center in its entirety. But even if for an $v \in V, c \in C$ it holds that $x_v^c \geq 1$ there may still exist $(v,c)$-cuts not containing any nodes fully assigned to $c$. To deal with this problem we need some results shown for node weighted Steiner trees.

\begin{definition}[Node Weighted Steiner Tree Problem]
    In an instance of the \emph{node weighted Steiner tree problem} we are given a graph $G = (V,E)$, a set of terminal nodes $T \subseteq V$ and a weight function $w: V \rightarrow \mathbb{R}_{\geq 0}$. The goal is to find a subset $S \subseteq V$ such that $T \subseteq S$, the nodes of $S$ form a connected subgraph of $G$, and the sum $\sum_{v \in S} w(v)$ is minimized. 
\end{definition}

Similar as for the connected $k$-median problem there also exists an LP relaxation for the node weighted Steiner tree problem:

\begin{alignat}{3}\label{LP_steiner}
    \min \quad && \multicolumn{2}{c}{ $\displaystyle\sum\limits_{v \in V} w(v) x_v$ } \nonumber \\
    \textrm{s.t.} \quad && \sum\limits_{v' \in S} x_{v'}&\geq 1, && \quad\forall u,v \in T, S \in \mathcal{S}_{u,v},\nonumber\\
    &&x_v^c &\geq 0, && \quad\forall v \in V,
\end{alignat}


One might note that the assignment version of the connected $k$-median problem is closely related to the node weighted Steiner tree problem as one could reinterpret the problem as finding $k$ node weighted Steiner trees such that each node is contained in at least one Steiner tree. The distances of the nodes to the respective centers can in this context be interpreted as $k$ different node weight functions for the different trees.
Using the solution $x$ we will now define for any center $c \in C$ the set $T_c := \{c\} \cup \{ v \mid x_v^c \geq 1 \}$. Then the solution $x$ can be subdivided into $k$ solutions for the Steiner tree LP using the respective sets $\left(T_c\right)_{c \in C}$ as terminals.

\begin{lemma}
\label{lem:split_lp}
    Let $(x_v^c)_{v \in V,c \in C}$ be a solution for LP \eqref{LP_nd_assignment}. Then for any center $c \in C$ the values $(x_v^c)_{v \in V}$ form a solution for LP \eqref{LP_steiner} with $T := T_c$ and $w(v) := d(v,c)$.
\end{lemma}

\begin{proof}
    This lemma is proved by contradiction.
    Assume that $(x_v^c)_{v \in V,c \in C}$ is a solution for LP \eqref{LP_nd_assignment} that violates some constraint of LP \eqref{LP_steiner}.
    Then there exist two terminals $u,v \in T_c$ and a $(u,v)$-cut $S$ such that
    \begin{equation*}
        \sum_{v' \in S} x_{v'}^c < 1.
    \end{equation*}

    Now we distinguish two cases. In the first case let $c$ be one of the nodes $\{u, v\}$, w.l.o.g. $u = c$. Then we can use the definition of LP \eqref{LP_nd_assignment} and $T_c$ to obtain
    \begin{equation*}
        \sum_{v' \in S} x_{v'}^c \geq x_v^c = 1,
    \end{equation*}
    which contradicts the assumption.

    In the second case both $u \neq c$ and $v \neq c$. Then by the previous argument $S$ can neither separate $u$ nor $v$ from $c$ because then its weight would be at least $1$. Thus there exists a path from $u$ to $c$ and from $c$ to $v$ not containing any nodes in $S$. But then there also exists a path from $u$ to $v$ not using any nodes in $S$ and $S$ would not be a $(u,v)$-cut.
\end{proof}

This lemma is crucial because for the node weighted Steiner tree problem there exists an $O(\log n)$-approximation algorithm by Klein and Ravi \cite{Klein.Ravi.1995} for which Guha et al.~\cite{GuMoNaSc99} later proved that it is implicitly a primal-dual algorithm and that the cost of the solution is at most $O(\log n)$ times worse than any fractional solution of LP \eqref{LP_steiner}. Alternatively Könemann et al provided a more detailed primal-dual algorithm for a generalized version of the node weighted Steiner tree problem ensuring the same guarantees \cite{konemann2013lmp}. Using these algorithms and Lemma \ref{lem:split_lp} we are able to subdivide $x$ into $k$ Steiner tree LP solutions that can each be transformed into a feasible cluster.

\AssignmentNonDisjoint*

\begin{proof}
    In the first step of the algorithm we will first calculate an optimum solution $\tilde{x}$ for the relaxed version of LP \eqref{LP_nd_assignment} which as argued above can be done in polynomial time. Since every assignment solution for the connected $k$-median problem is also an LP solution, the costs of $\tilde{x}$ are upper bounded by the optimum value \textrm{OPT}. After multiplying the values by $k$, the cost of solution $x$ is then again bounded by $k\cdot \textrm{OPT}$.

    Consider $T_c$ as defined above, by Markov's inequality for every node $v$ there exists a center $c$ such that $\tilde{x}_v^c \geq \frac{1}{k}$ which means that $v \in T_c$. Thus we have $\bigcup_{c \in C} T_c = V$. By Lemma \ref{lem:split_lp} the fractional solution $x$ can be split up into $k$ solutions of the Steiner tree LP \eqref{LP_steiner} with the respective terminal sets $\left(T_c\right)_{c \in C}$ such that the total sum of the objective value is exactly the cost of $x$. 

    By applying the algorithm by Klein and Ravi \cite{Klein.Ravi.1995} with $T = T_c$ and $w(v) = d(v,c)$ for every $c \in  C$ separately, one obtains a Steiner tree $S_c$ with $T_c \subseteq S_c$ whose cost is upper bounded by $O(\log n)$ times the cost of the respective LP solution. Thus one obtains in polynomial time $k$ Steiner trees each connecting the respective terminal set $T_c$ and total cost bounded by $O(\log n)$ times the objective value of $x$ which results in an upper bound of $O\left(k \log n\right) \textrm{OPT}$. Since every node is contained in at least one of these Steiner trees, this directly yields the desired solution for the assignment version of the non-disjoint connected $k$-median problem.

    This proof works even if the distances are not a metric because it does not use any of their properties.
    %
\end{proof}

One may note that this approach works with any approximation algorithm for the node weighted Steiner tree problem whose cost is bounded by the optimum LP solution. For some restricted graph classes, including planar graphs, Demaine et al. \cite{demaine2009node} provided a primal dual algorithm calculating a constant approximation. Thus the approach can also be used to get an $O(k)$-approximation (and $O(k^2)$ if the centers are not given) for these kind of connectivity graphs.

%% file: nd_centers.tex
In this section we give an overview how suitable centers for the non-disjoint connected $k$-median problem can be found. For a more detailed description of the procedure and its analysis we refer to Section \ref{apx:finding_centers}. First we will again provide a suitable ILP formulation:  

\begin{alignat}{3}\label{LP_cut_centersMain}
    \min \quad && \multicolumn{2}{c}{ $\displaystyle\sum\limits_{v,c \in V} d(v,c) x_v^c$ } \nonumber\\
    \textrm{s.t.} \quad &&\sum\limits_{c \in V} x_v^c & \geq 1, && \quad\forall v \in V,\nonumber\\ 
    &&\sum\limits_{c \in V} x_c^c &\leq k, && \nonumber\\
    &&\sum\limits_{v' \in S} x_{v'}^c &\geq x_{v}^c, && \quad\forall v,c \in V, S \in \mathcal{S}_{v,c},\nonumber\\
    &&x_v^c &\in \{0,1\}, && \quad\forall v,c \in V,
\end{alignat}


where $\mathcal{S}_{v,c}$ denotes the set of vertex cuts, separating $v$ and $c$, defined as in the previous section.

The main difference to the assignment variant is that we do now allow to assign nodes to any other nodes and not only to some predefined centers. However, by definition of $\mathcal{S}_{v,c}$, it holds that $\{c\} \in \mathcal{S}_{v,c}$ for any $v,c \in V$. Thus $x_v^{c}$ will automatically be limited by $x_{c}^{c}$. This naturally leads to the property that nodes can only be assigned to $c$ if $x_{c}^{c}= 1$. Thus limiting $\sum_{c \in V} x_c^c$ by $k$ directly ensures that any solution of the ILP uses at most $k$ centers.

If we relax the integrality constraint and only require that $x_v^c \geq 0$ for all $v,c \in V$ also LP \eqref{LP_cut_centersMain} can be solved in polynomial time despite its exponential size. The drawback of this is that for arbitrarily many nodes $c' \in V$ it can hold that $x_{c'}^{c'} >0$ if the respective values are sufficiently small. This means that other nodes can at least partially be assigned to them and we will call the value $x_{c'}^{c'}$ the \emph{opening} of $c'$.

The main focus of this section is to move these openings to a limited set of $k$ centers that will be determined by our algorithm such that reassigning the nodes accordingly is not too expensive. To do this, we follow a similar approach as the LP rounding algorithm for $k$-median by Charikar et al. \cite{charikar1999constant}. In the first step, we will find a solution in which every node is opened either by at least $\frac{1}{2}$ or not opened at all (limiting the number of centers to $2k$). In the latter steps these openings will then be made integral by first finding possible replacements for some of the centers, which possibly allows us to close the respective centers, and then choosing a subset of the centers of size $k$ such that for every center itself or its replacement is contained in this subset.

We will relax the connectivity constraint during the algorithm. We say for an $\alpha \in [0,1]$ that an assignment $\left(x_v^c\right)_{v,c \in V}$ is an $\alpha$-connected solution if $\sum_{c \in V} x_c^c \leq k$, $\forall v \in V: \sum_{c \in V} x_v^c \geq 1$ and for all $v,c \in V$ and all cuts $S \in \mathcal{S}_{v,c}$ separating $v$ and $c$ it holds that $\sum_{v' \in S} x_{v'}^c \geq \min(\alpha,x_v^c)$. To simplify notation, we will denote the weight of the minimum cut between two nodes $v$ and $c$ (with respect to the values $\left(x_{v'}^{c}\right)_{v' \in V}$) by $\sep^x(v,c) = \min_{S \in \mathcal{S}_{v,c}}\sum_{v' \in S} x_{v'}^c$. Then the last requirement can be reformulated as $\sep^x(v,c) \geq \min(\alpha,x_v^c)$. It is easy to verify that an assignment corresponds to a fractional solution for LP \eqref{LP_cut_centersMain} if and only if it is $1$-connected. 

Now we give an overview over the different steps of the algorithm. Our main focus will lie on the first step, which ensures that all centers will be half opened afterwards. This is due to the fact that it introduces the largest amount of novel ideas and necessary modifications compared to the existing algorithm by Charikar et al.:

\textbf{Step 1:} To obtain half-opened centers, let an optimal (fractional) LP solution $\tilde{x}$ be given. For each $v$, let its mass be defined as $m_v = \sum_{c \in V} \tilde{x}_v^c$, its cost (or price) as $p_v = \sum_{c \in V} \tilde{x}_v^c \cdot d(v,c)$ and its (average) radius as $r_v = \frac{p_v}{m_v}$. We say that an assignment of a node $v$ to another node $c$ is \emph{good} if $d(v,c) \leq 4r_v$. Since $m_v \geq 1$ for all nodes $v \in  V$, we may observe that because of Markov's inequality at least $\frac{3}{4}$ of $v$ gets assigned via good assignments, i.e., $\sum_{c \in V: d(v,c) \leq 4r_v} \tilde{x}_v^{c} \geq \frac{3}{4}$. 

We will simultaneously maintain two assignments $(x_v^c)_{v,c \in V}$ and $(y_v^c)_{v,c \in V}$ such that for all $v,c  \in V$ initially $x_v^c = \tilde{x}_v^c$ and $y_v^c = 0$ and that at the end of the algorithm $x_v^c = 0$. The final values of $y_v^c$ will be the assignment using half opened centers. During the algorithm, the combination of $x$ and $y$ will always be a $\frac{1}{8k}$-connected solution, which means that the following inequalities stay satisfied throughout the entire process:

\begin{align}
         \sum_{c \in V} \left(x_v^c + y_v^c \right)& \geq 1 & \forall v \in V \label{half_inv_assignedMain}\\
         \sum_{c \in V} \left(x_c^c + y_c^c\right)&\leq k& \label{half_inv_numberMain}\\
         \sum_{v' \in S} x_{v'}^c &\geq x_{v}^c & \forall v,c \in V, S \in \mathcal{S}_{v,c} \label{half_inv_xMain}\\
         \sum_{v' \in S} y_{v'}^c &\geq \min\left(\frac{1}{8k}, y_{v}^c\right) &  \forall v,c \in V, S \in \mathcal{S}_{v,c} \label{half_inv_yMain}
\end{align}

The algorithm will now maintain an initially empty center set $C$ and repeat in an alternating order the following two steps: First it decides for an additional node whether it gets added to the center set or not (in the latter case it will be added to a so called environment of an already existing center). Afterwards it shifts some openings in $x$ (the $x_c^c$ that are greater than $0$) to already chosen centers $c$ in $C$.

To get an intuition regarding the shifts, it can be beneficial to sometimes think about the connectivity constraints not from a cut-perspective but to rather interpret the assignments as flows. Using the max-flow min-cut theorem, it is easy to see that if we interpret the values $(x_{v'}^{c'})_{v' \in V}$ as node capacities then there exists a flow from a node $v$ to a node $c'$ of value $f$ exactly if $\sep^x(v,c') \geq f$. Thus our connectivity constraint implies that there always exists a flow from $v$ to $c'$ of value $x_v^{c'}$. By basically reverting this flow for a center $c \in C$ and a node $c' \in V$ we are able to shift an opening of value $x_c^{c'}$ from $c'$ to  $c$. One important question is now how the nodes that were previously assigned to $c'$ get assigned. If $x_c^{c'}$ was equal to $x_{c'}^{c'}$, the entire opening of $c'$ can be shifted and we can simply reassign all nodes to $c$ by setting $y_v^c = y_v^c + x_v^{c'}$ and $x_v^{c'} = 0$ for any $v \in V$. However, if $x_c^{c'}$ was smaller than $x_{c'}^{c'}$ only a fraction of the opening of $c'$ gets shifted and we basically have to split up the previous assignments to $c'$ between $c$ and $c'$. As Figure \ref{fig:split_opening} shows this can lead to situations in which the total assignment to $c$ and $c'$ after this shift actually needs to be larger than the assignment before this shift. Since it is actually rather involved to determine the new values in this situation, a formal description will only be provided in Appendix \ref{sec:half}. To get an intuition regarding the algorithm it is for now sufficient to only consider the first case. 

\begin{figure}
\centering
    \includegraphics[width= 0.4 \textwidth]{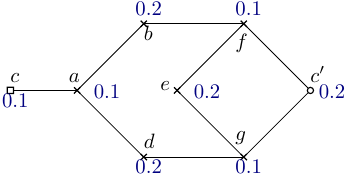}
    \caption{A sketch depicting a situation in which, at least if we require $1$-connectivity, the total assignment needs to increase if we apply a shift from $c'$ to $c$. The blue values denote the assignments of the respective nodes to $c'$ before the shift, i.e., the $x_v^{v'}$-values. We assume that $x_v^c=0$ for all nodes $v$. Now if we apply a shift from $c'$ to $c$ by $x_c^{c'}=0.1$, both $c'$ and $c$ have an opening of $0.1$. Thus the nodes $b,d,e$ all need to be assigned to $c$ and $c'$ with value 0.1. If we would try to find assignments of $a$, $f$ and $g$ such that for each node its total assignment to $c$ and $c'$ is exactly $0.1$ we would need to set $x_g^{c'} = x_f^{c'} = 0.1$ to ensure the connectivity of $b$ and $d$ to $c'$. However this would mean that $y_f^c = y_g^c = 0$ which would separate $e$ from $c$. Thus to ensure connectivity, for at least one node the total assignment after the shift needs to be larger than the assignment before the shift.}
    \label{fig:split_opening}
\end{figure}

Since either way also the assignments of nodes assigned to $c'$ have to be shifted in this situation, we have to make sure that the distances of the respective shifts are not too large which the algorithm will ensure by only allowing shifts with a rather small distance between $c$ and $c'$ in the beginning and increasing this distance over time. More details will follow later. 

While the shifts depend on the chosen centers $C$, the choice whether a node $v$ will be added to $C$ in turn depends on the shifts that have already been performed before $v$ gets considered by the algorithm. The algorithm terminates once it has decided for every node whether it gets added to $C$ and all openings in $x$ have been shifted.

To determine which nodes should be added to $C$, we iterate over all nodes ordered by their radii in increasing order and check whether we can open them by at least $\frac{1}{2}$. To be more precise, a node $v$ will be chosen as a center if the sum of the remaining good assignments $(x_v^{c'})_{c' \in V}$ is at least $\frac{1}{2}$. The construction of the algorithm will ensure that it will shift at least all remaining openings belonging to good assignments of $v$ to $v$ itself which means that $v$ will be at least half opened. When $v$ is added to the center set we also create its environment $S_v = \{v\}$, which initially only contains $v$ itself.

In the case that $v$ does not get chosen as a center, we know that there were at least $\frac{1}{4}$ of good assignments of $v$ that have already been shifted to other centers in $C$. Since every center in $C$ gets at least half opened and the entire opening is limited by $k$, we know that $|C| \leq 2k$ and thus there exists at least one center to which $\frac{1}{8k}$ of the good assignments of $v$ have been shifted.
One critical property of the algorithm will be that before $v$ gets considered, the distance between the source $c'$ and target $c$ of any shift will be at most $4r_v$. Thus for any good assignment to a node $z$ whose opening has been shifted to a center $c$ we have that $d(v,c) \leq d(v,z) + d(z,c) \leq 8r_v$.
As a result we may conclude that there exists a center $c \in C$ such that $y_v^c \geq \frac{1}{8k}$ while $d(v,c) \leq 8 r_v$ and we add $v$ to the environment $S_c$ of the closest center $c$ fulfilling these requirements.

Adding $v$ to this environment allows us to shift some openings of a node $c'$  with $x_v^{c'}> 0$ over the node $v$ to $c$ even if there does not exist a direct flow from $c$ to $c'$ (i.e., $c$ is not assigned to $c'$). The intuition behind this is that we first revert the flow from $v$ to $c'$ and thus are already shifting the opening to $v$. Now by combining $y_v^c \geq \frac{1}{8k}$ and Property \eqref{half_inv_yMain} we know that already before this shift the minimum cut between $v$ and $c$ has value at least $\frac{1}{8k}$. Thus if we are moving the opening (and the assignments to it) to $c$ instead of $v$ itself we know that the solution still stays $\frac{1}{8k}$-connected. Furthermore since $v$ and $c$ are not too far apart, we know that the distances of these shifts does not increase too much compared to a direct shift to $v$. This will help us ensuring that if the openings of some nodes, to which $v$ got assigned, get shifted to some other centers then the distance added by this shift will not be too large.

Each time after a new node $v$ gets added to either $C$ or to an environment $S_c$ of a center $c \in C$ we will go over all possible shifts ordered by their distances in increasing order, and perform all shifts with distance at most $4 r_{\mathrm{next}}$, where $r_{\mathrm{next}}$ denotes the radius of the next node that will be considered by the algorithm. Note that any shift with radius $ \leq 4_{r_v}$ can only result from the addition of $v$ to $S_c$ or $C$ because any shift not involving $v$ would already have been applied in the last iteration (as $r_v$ is the previous value of $r_{next}$). As a result all those shifts will be performed before any shift not involving $v$ will be executed. This directly guarantees that if $v$ is added to $C$ then all openings belonging to good assignments will be shifted to $v$, which ensures that $v$ is half opened.

An important property of the shifts is that if we shift some opening of a node $c' \in V$ over a node $v \in V$ then afterwards $x_v^{c'} = 0$ which ensures that for all $v,c' \in V$ we will only shift some opening of $c'$ over $v$ at most once, limiting the number of shifts by $|V|^2$. Additionally after we have added $v$ to $C$ or a set $S_c$ there always exists a possible shift while $x_v^{c'} > 0$ for a $c' \in V$. Thus, it is possible to execute the described algorithm for a polynomial time until for every $v,c' \in V$, $x_v^{c'} = 0$ and then return $\left(y_v^{c}\right)_{v \in V, c \in C}$. A pseudocode summarizing the structure of the described procedure can be found in Algorithm \ref{Alg:pseudo_half}. The cost of the resulting $\frac{1}{8k}$-connected solution is bounded by $O(k)$ times the cost of the LP solution $\tilde{x}$.

\begin{algorithm}[H]
	\DontPrintSemicolon
	\SetKwInOut{Input}{input}
	\SetKwInOut{Output}{output}
    \caption{The general structure of the first step.}
    \label{Alg:pseudo_half}
	Order $v_1,...,v_n$ according to their radii $(r_v)_{v \in V}$ in increasing order\;
    \For{$v_1,...,v_n$}{
        Decide whether $v_i$ is added to the center set $C$ or to an environment $S_c
        $ of a center $c \in C$\;
        \While{A shift from a node $v$ to a center $c \in C$ with $d(v,c) \leq r_{v_{i+1}}$ is possible}{
            Apply the cheapest shift\;
        }
    }
\end{algorithm}

\textbf{Step 2:} Now we would like to reduce the number of centers to at most $k$. However, if we want to reduce the opening of a center $c \in C$ to $0$, we will need to find a suitable replacement, i.e., another center $\tilde{c}$ such that reassigning all nodes assigned to $c$ to $\tilde{c}$ neither costs too much nor violates our connectivity constraints (which is further relaxed to $\frac{1}{16k}$-connectivity). To do this we will need two properties of Algorithm \ref{Alg:pseudo_half}. The respective proofs can be found in Section \ref{sec:replace}:
\begin{itemize}
    \item For any node $v$ we know that if a shift from another node $c'$ to a center $c$ caused $v$ to be (partially) assigned to $c$ then the radius $r_c$ of $c$ can be bounded by the radius $r_v$ of $v$ and the distance from $v$ to $c'$. To be more precise $r_c \leq d(v,c') + 2 r_v$. Intuitively, if this was not the case, then the algorithm would have shifted some opening of $c'$ either to $v$ itself (if $v$ gets chosen as a center) or over $v$ to another center $\tilde{c} \in C$ with $v \in S_{\tilde{c}}$ before $c$ was even added to the center set which would have prevented $v$ from being assigned to $c$. The main consequence of this property is that if we can find another center $\tilde{c}$ that is only a constant times $r_c$ away from $c$ then reassigning all nodes that are currently assigned to $c$, to $\tilde{c}$ instead costs not too much.
    \item Additionally for any center $c$ it holds that if $c$ got assigned to another center $\tilde{c}$ because some opening got shifted from a node $c'$ to $c$ then it holds that $d(c,\tilde{c})$ is bounded by $d(c,c') + \max(4r_c,d(c,c'))$ because otherwise a shift from $c'$ to $c$ would have happened before the shift to $\tilde{c}$. Since the average length of the original assignments of $c$ was $r_c$, one can use this to identify a center $\tilde{c}$ that is not too far apart from $c$.
\end{itemize}

Using these two properties, we are able to split up $C$ into two sets $C_1$ and $C_{1/2}$ such that for every center $c \in C_{1/2}$ we have a successor $s(c) \in C$ such that the total cost that occurs if we replace any center in $C_{1/2}$ by its successor would not be too large and $|C_1| + \frac{1}{2}|C_{1/2}| \leq k$. Later we will then open all nodes in $C_1$ as well as at most half the nodes in $C_{1/2}$ to obtain a integral center set $C^*$ such that for all nodes in $C_{1/2}$ either the node itself or its successor is contained in $C^*$.

To obtain these sets, we consider an arbitrary center $c \in C$ and define $a_c = 1 - y_c^c$.  Then at least a total value of $a_c$ of the openings corresponding to the initial assignments of $c$ (the values $\tilde{x}_c^{c'}$ for any $c' \in V$) were shifted to other centers. If $a_c$ has value at least $\frac{1}{4}$, we are able to use Markov's inequality together with the second property to identify another center $\tilde{c}$ such that $d(c,\tilde{c})$ lies within $O(r_c)$ and $c$ is already sufficiently connected to $\tilde{c}$ (i.e., $y_c^{\tilde{c}} \geq \frac{1}{16k}$). We thus can set $s(c) = \tilde{c}$ and add $c$ to $C_{1/2}$.

However, if $a_c$ is rather small then it can be that all assignments of $c$ to other centers $\tilde{c}$ only resulted from shifts of nodes $c'$ to $\tilde{c}$ that are really far apart from $c$ which means that $\tilde{c}$ can be even further away. Additionally the assignments $y_c^{\tilde{c}}$ might also be quite small which means that we would need to increase some existing assignments to $\tilde{c}$ to ensure $\frac{1}{16k}$-connectivity if we replace $c$ by $\tilde{c}$. However, if there are multiple centers where the respective value $a_c$ is small then we know that since $\sum_{c \in C} (1- a_c) = \sum_{c \in C} y_c^c \leq k$ that if we add one of these centers to $C_{1/2}$ we can add multiple others to $C_1$. Since we do not need a replacement for the centers in $C_1$, we can use a careful potential argument to redistribute parts of the replacement cost of the centers added to $C_{1/2}$ to the centers in $C_1$ to bound the additional cost among all nodes.

\textbf{Step 3:} To decide which centers in $C_{1/2}$ can be closed we have a look at the directed Graph $G_{1/2} = (C_{1/2},S)$, where $S = \left\{ (c,s(c))\mid c,s(c) \in C_{1/2}\right\}$ contains for any center in $C_{1/2}$ an edge to its successor if the latter is again contained in $C_{1/2}$. If this graph is bipartite we can simply take the smaller half of this partition and add this to the center set $C^*$. Then for every edge one of its two endpoints is contained in $C^*$ which ensures that for every center either itself or its successor is contained in $C^*$. However, the graph might not directly be bipartite because it can contain cycles of uneven length. But we can ensure by being a little bit more careful in step 2 that any of these cycles contains at most one center $c$ with $a_c < \frac{1}{4}$. Using this, we can remove the respective cycles and thus obtain a bipartite graph. For all centers in $C_{1/2}$ not added to $C^*$, we reassign all nodes assigned to them to the respective successor $s(c)$. Since we bounded the cost of replacing all centers in $C_{1/2}$ and are only replacing a subset of them, the total cost of this is again bounded. The entire cost still stays in $O(k)$ times the cost of the fractional LP solution $\tilde{x}$.

\textbf{Step 4:} After step 3, we have an $\frac{1}{16k}$-connected solution only using integral centers. However the assignments might still be fractional. By multiplying the assignments by $16k$ the solution becomes $1$-connected and similarly as for the assignment version in Section \ref{sec:NDCkMAP} we obtain that for every node $v$ there exists at least one center $c$ such that $x_v^c \geq 1$. Again applying the primal dual algorithm for node weighted Steiner trees, we obtain an entirely integral solution:

\ClusteringNonDisjoint*

%% file: integral_centers.tex
\section{Making the centers integral}\label{apx:finding_centers}

In this section, we give a complete description of how suitable centers for the non-disjoint $k$-median problem can be found. Since we wanted this description to be self-contained, some parts of it will be repetitions of Section \ref{sec:nd_find_centers_short}. We first provide a suitable LP formulation:  

\begin{alignat}{3}\label{LP_cut_centers}
    \min \quad && \multicolumn{2}{c}{ $\displaystyle\sum\limits_{v,c \in V} d(v,c) x_v^c$ } \nonumber \\
    \textrm{s.t.} \quad&&\sum\limits_{c \in V} x_v^c & \geq 1, && \quad\forall v \in V,\nonumber\\
    &&\sum\limits_{c \in V} x_c^c &\leq k, && \nonumber\\
    &&\sum\limits_{v' \in S} x_{v'}^c &\geq x_{v}^c, && \quad\forall v,c \in V, S \in \mathcal{S}_{v,c},\nonumber\\
    &&x_v^c &\in \{0,1\}, && \quad\forall v,c \in V,
\end{alignat}
where $\mathcal{S}_{v,c}$ denotes the set of vertex cuts, separating $v$ and $c$, defined in Definition \ref{def:cuts}.

The main difference to the assignment variant is that we do now allow to assign nodes to any other nodes and not only to some predefined centers. However by definition of $\mathcal{S}_{v,c}$ it holds that $\{c\} \in \mathcal{S}_{v,c}$ for any $v,c \in V$. Thus, $x_v^{c}$ will automatically be limited by $x_{c}^{c}$. This naturally leads to the property that nodes can only be assigned to $c$ if $x_{c}^{c}= 1$. Thus limiting $\sum_{c \in V} x_c^c$ by $k$ directly ensures that any solution of the LP uses at most $k$ centers.

If we relax the integrality constraint and only require that $x_v^c \geq 0$ for all $v,c \in V$ also LP \eqref{LP_cut_centers} can be solved in polynomial time despite its exponential size. The drawback of this is that arbitrarily many nodes $c' \in V$ can fulfill that $x_{c'}^{c'} >0$ if the respective values are sufficiently small. This means that other nodes can at least partially be assigned to them and we will call the value $x_{c'}^{c'}$ the \emph{opening} of $c'$.

The main focus of this section is to move these openings to a limited set of $k$ centers that will be determined by our algorithm such that reassigning the nodes accordingly is not too expensive. To do this, we follow a similar approach as the LP rounding algorithm for $k$-median by Charikar et al. \cite{charikar1999constant}. In the first step we will find a solution where every node is opened either by at least $\frac{1}{2}$ or not opened at all (limiting the number of centers to $2k$). In the latter steps these will then be made integral by first finding possible replacements for some of the centers if we choose to reduce the respective opening to $0$ and then choosing a subset of the centers of size $k$ such that for every center itself or its replacement is contained in this subset.

However during this we have to be very careful with our connectivity constraints. Because of this we need a more general definition of vertex cuts as well as some additional notation:

\begin{definition}\label{def:cuts}
For a set $N \subseteq V$ we say that $N$ is a \emph{vertex cut} between two sets  $S \subseteq V$ and $T \subseteq V$ if every path between two nodes $v \in S$ and $u \in T$ contains at least one node in $N$. One might note that by this definition $N$ can contain nodes from $S$ and $T$. We define the interior $I_T(N)$ of the cut $N$ as the set of all points $v \in V\setminus N$ for which no path from $v$ to any node in $T$ exists in $G \setminus N$. Additionally we set $H_T(N)= N \cup I_T(N)$. Then $N$ is a cut between $S$ and $T$ iff $S \subseteq H_T(N)$.

For a weight function $w: V \rightarrow \mathbb{R}_{\geq0}$ the weight of the cut $N$ is defined as $w(N) = \sum_{v \in N} w(v)$ and a cut between $S$ and $T$ is called a minimum (weight) vertex cut if it minimizes $w(N)$. Its respective weight is also denoted by $\sep^w(S,T)$. For any $v \in V$ the increase of $\sep^w(S,T)$ if we add $v$ to $S$ is denoted by 
\begin{equation*}
\Delta^w(S,v,T) = \sep^w(S \cup \{v\},T) - \sep^w(S,T).
\end{equation*}
\end{definition}

For a matrix $(x_v^c)_{v,c \in V}$ and for any $S \subseteq V$, $c \in V$ we can define $w^c_x(v) = x_v^c$ for all $v \in V$ and with a slight abuse of notation we will use $\sep^x(S,c) = \sep^{w_x^c}(S,\{c\})$ and $\Delta^x(S,v,c) = \Delta^{w_x^c}(S,v,\{c\})$. One may observe that using this notation we may reformulate the LP constraints $ \forall S \in \mathcal{S}_{v,c}: \sum_{v' \in S} x_{v'}^c \geq x_{v}^c $ for any $v,c \in V$ as $\sep^x(v,c) \geq x_v^c$.

Additionally we will relax the connectivity constraint during the algorithm. We say for an $\alpha \in [0,1]$ that an assignment (actually a matrix) $\left(x_v^c\right)_{v,c \in V}$ is an $\alpha$-connected solution if $\sum_{c \in V} x_c^c \leq k$, $\forall v \in V: \sum_{c \in V} x_v^c \geq 1$ and for all $v,c \in V$ it holds that $\sep^x(v,c) \geq \min(\alpha,x_v^c)$. It is easy to verify that an assignment is a solution for LP \eqref{LP_cut_centers} if and only if it is $1$-connected.

\subsection{Obtaining half opened centers}
\label{sec:half}

To obtain half-opened centers, let an optimum LP solution $\tilde{x}$ be given. For each $v$, let its mass be defined as $m_v = \sum_{c \in V} \tilde{x}_v^c$, its cost (or price) as $p_v = \sum_{c \in V} \tilde{x}_v^c d(v,c)$ and its (average) radius as $r_v = \frac{p_v}{m_v}$. We say that an assignment of a node $v$ to $c$ is \emph{good} if $d(v,c) \leq 4r_v$. Since $m_v \geq 1$ for all nodes $v \in  V$, we may observe that because of Markov's inequality at least $\frac{3}{4}$ of $v$ gets assigned via good assignments, i.e., $\sum_{c \in V: d(v,c) \leq 4r_v} \tilde{x}_v^{c} \geq \frac{3}{4}$. 

We will simultaneously maintain two assignments $(x_v^c)_{v,c \in V}$ and $(y_v^c)_{v,c \in V}$ such that for all $v,c$ initially $x_v^c = \tilde{x}_v^c$ and $y_v^c = 0$ and that at the end of the algorithm $x_v^c = 0$. The final values of $y_v^c$ will be the assignment using half opened centers. During the algorithm, the combination of $x$ and $y$ will always be a $\frac{1}{8k}$-connected solution, which means that the following inequalities stay fulfilled:

\begin{align}
         \sum_{c \in V} x_v^c + y_v^c & \geq 1 & \forall v \in V \label{half_inv_assigned}\\
         \sum_{c \in V} x_c^c + y_c^c&\leq k& \label{half_inv_number}\\
         \sum_{v' \in S} x_{v'}^c &\geq x_{v}^c & \forall v,c \in V, S \in \mathcal{S}_{v,c} \label{half_inv_x}\\
         \sum_{v' \in S} y_{v'}^c &\geq \min\left(\frac{1}{8k}, y_{v}^c\right) &  \forall v,c \in V, S \in \mathcal{S}_{v,c} \label{half_inv_y}
\end{align}

The algorithm will now maintain an initially empty center set $C$ and repeats in an alternating order the following two steps: First it decides for an additional node whether it gets added to the center set or not (in the latter case it will be added to a so called environment of an already existing center). Afterwards it shifts some openings in $x$ (the $x_c^c$ that are greater than $0$) to already chosen centers $c$ in $C$. 

To get an intuition regarding the shifts it can be beneficial to sometimes think about the connectivity constraints not from a cut-perspective but rather interpret the assignments as flows. Using the min-flow max-cut theorem it is easy to see that if we interpret the values $\left(x_{v'}^{c'}\right)_{v' \in V}$ as node capacities then there exists a flow from a node $v$ to a node $c'$ of value $f$ exactly if $\sep^x(v,c') \geq f$. Thus our connectivity constraint implies that there always exists a flow from $v$ to $c'$ of value $x_v^{c'}$. By basically reverting this flow for a center $c \in C$ and a node $c' \in V$ we are able to shift an opening of value $x_c^{c'}$ from $c'$ to  $c$. Since also the assignments of nodes assigned to $c'$ have to be shifted in this situation, we have to make sure that the distances of the respective shifts are not too large. More details on the exact shift procedure as well as which shifts get performed will follow later. While the shifts are obviously depending on the chosen centers $C$ the choice whether a node $v$ will be added to $C$ is in turn heavily depending on the shifts that were already performed before $v$ gets considered by the algorithm. The algorithm terminates once it has decided for every node whether it gets added to $C$ and all openings in $x$ have been shifted.

To determine which nodes should be added to $C$ we iterate over all of them ordered by their radii in increasing order and check whether we can open them by at least $\frac{1}{2}$. To be more precise a node $v$ will be chosen as a center, if the sum of the remaining good assignments $(x_v^{c'})_{c' \in V}$ is at least $\frac{1}{2}$. The construction of the algorithm will ensure that it will shift at least all remaining openings belonging to good assignments of $v$ to $v$ itself which means that $v$ will be at least half opened. 

In the case that $v$ does not get chosen as a center we know that there were at least $\frac{1}{4}$ of good assignments of $v$ that have already been shifted to one of the centers in $C$. Since every center in $C$ gets at least half opened and the entire opening is limited by $k$ we thus know that $|C| \leq 2k$ and thus there exists at least one center to which $\frac{1}{8k}$ of the good assignments of $v$ have been shifted.
One critical property of the algorithm will be that before $v$ gets considered, the distance between the source $c'$ and target $c$ of any shift will be at most $4r_v$. Thus for any good assignment to a node $z$ whose opening has been shifted to a center $c$ we have that $d(v,c) \leq d(v,z) + d(z,c) \leq 8r_v$.
As a result we may conclude that there exists a center $c \in C$ such that $y_v^c \geq \frac{1}{8k}$ while $d(v,c) \leq 8 r_v$ and we add $v$ to the environments $S_c$ of the closest center $c$ fulfilling these requirements.
Initially for every center $c \in C$ the environment is set to $S_c = \{c\}$ when $c$ gets added to $C$.

Adding $v$ to this environment allows us to shift some openings of a node $c'$  with $x_v^{c'}> 0$ over the node $v$ to $c$ even if there does not exist a direct flow from $c$ to $c'$ (i.e. $c$ is not assigned to $c'$). The intuition behind this is that we first revert the flow from $v$ to $c'$ and thus are already shifting the opening to $v$. Now by combining $y_v^c \geq \frac{1}{8k}$ and Property \eqref{half_inv_y} we know that already before this shift the minimum cut between $v$ and $c$ has value at least $\frac{1}{8k}$. Thus if we are moving the opening (and the assignments to it) to $c$ instead of $v$ itself we know that the solution still stays $\frac{1}{8k}$-connected. Furthermore since $v$ and $c$ are not too far apart, we know that the distance of these shifts does not increase too much compared to a shift directly to $v$. This will help us ensuring that if the openings of some nodes, to which $v$ got assigned, get shifted to some other centers then the distance added by this shift will not be too large.

Now let us go over the formal details of the shifts. To determine the values of $\left(x_v^c\right)_{v,c \in V}$ we will maintain  for any $c' \in V$ a set $D_{c'}$ which is the set of all nodes over which already some of the opening of $c'$ has been shifted. Then at any point during the algorithm for all $v,c' \in V$ it holds that $x_v^{c'} = \Delta^{\tilde{x}}(D_{c'},v,c')$, i.e. $x_v^{c'}$ is the amount by how much the minimum cut between $D_{c'}$ and $c'$ increases if we add $v$ to $D_{c'}$. We say that a shift from a node $c' \in V$ to a node $c \in C$ is possible if there exists a $v \in S_c$ (one of the nodes over which we may shift opening to $c$) such that $x_v^{c'} > 0$. 

To get a better intuition regarding $D_{c'}$ as well as the choice of the $x$-values, the flow-perspective might again be helpful. As mentioned above we can view the shift of some opening from $c'$ over a node $v$ to some center $c$ as directing a flow of values $x_v^{c'}$ from $c'$ to  $v $ (which then will be further shifted to $c$). Now if $x_v^{c'}$ is smaller than $x_{c'}^{c'}$ further shifts will be necessary to move the entire opening away from $c'$. Since for each of these further shifts again a flow is send from $c'$ to the respective nodes we can also interpret this as having a singular combined flow from a single source $c'$ to multiple sinks where the sinks are all the nodes in the set $D_{c'}$. For connectivity reasons we would like that also for this combined flow, the amount passing through a node $v' \in V$ is bounded by $\tilde{x}_{v'}^{c'}$. Since cuts are symmetric, $\sep^{\tilde{x}}(D_{c'},c')$ is also the value of the minimum cut between $c'$  and $D_{c'}$ and thus the value of the maximum possible flow from $c'$ to $D_{c'}$. Following the same logic we may also observe that $x_v^{c'} = \Delta^{\tilde{x}}(D_{c'},v,c')$ is the amount of additional flow that we are getting if we add $v$ to the set of sinks $D_{c'}$. Over the course of the algorithm it will now add more and more nodes to $D_{c'}$ and shift as much opening away from $c'$ as possible until the value of the maximum flow from $c'$ to $D_{c'}$ is exactly $\tilde{x}_{c'}^{c'}$ which means that the entire opening of $c'$ has been shifted to some centers $c \in C$. Additionally as shown below, this update of the values of $x$ also ensures that the addition of $v$ to $D_{c'}$ will reduce the value of $x_{c'}^{c'}$ by $x_v^{c'}$ thus ensuring that Property \eqref{half_inv_number} is preserved. Figure \ref{fig:values_delta} gives an example how the values of $\left(x_v^{c'}\right)_{v,c' \in V}$ change when some nodes get added to $D_{c'}$.

\begin{figure}
\begin{subfigure}{0.45\textwidth}
    \includegraphics[width=\textwidth, page =1]{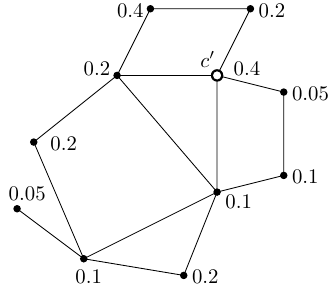}
    \subcaption{The initial values $x_v^{c'} = \tilde{x}_v^{c'}$, while $D_{c'} = \emptyset$.}
\end{subfigure}
\hfill
\begin{subfigure}{0.45\textwidth}
    \includegraphics[width=\textwidth, page =2]{example_delta.pdf}
    \subcaption{The values $x_v^{c'}$ for $D_{c'} = \{v_1,v_2,v_3\}$}
\end{subfigure}
\caption{An example of the assignments to a node $c' \in V$ for $D_{c'} = \emptyset$ and $D_{c'} = \{v_1,v_2,v_3\}$. One may note that the latter values would be the same if $v_1$ was not part of $D_{c'}$.}
\label{fig:values_delta}
\end{figure}

To simplify the analysis of the algorithm we will store for any assignment $y_v^c$ which shifts caused this assignment. To be more precise for any $v,c' \in V$ and $c \in C$ we have a value $y_v^{c',c}$ that tells us how much $v$ got assigned to $c$ by shifts from $c'$ to $c$. At any point the value $y_v^c$ is implicitly defined as $\sum_{c' \in V} y_v^{c',c}$. 

Each time after a new node $v$ gets either added to $C$ or to an environment $S_c$ of a center $c \in C$ we will go over all possible shifts ordered by their distances in increasing order, and perform all shifts with distance at most $4 r_{next}$ where $r_{next}$ denotes the radius of the next node that will be considered by the algorithm. One may note that any shift with a radius $ \leq 4_{r_v}$ can only result from the addition of $v$ to $S_c$ or $C$ and that all those shifts will be performed before any shift not involving $v$ will be executed. This directly guarantees us that if $v$ got added to $C$ that all openings belonging to good assignments will be shifted to $v$ which ensures that $v$ is half opened.

The first time we shift some opening from $c'$ to $c$ we will store the current values of $(x_v^{c'})_{v \in V}$ in $(u_v^{c',c})_{v \in V}$. Functionally for any $v,c'$ the value $u_v^{c',c}$ acts as an upper bound limiting the amount $x_v^{c',c}$ by which $v$ can be assigned to $c$ via shifts from $c'$ to $c$. Every time we perform a shift from $c'$ over $v$ to $c$ we will increase $y_{c}^{c',c}$ by $x_v^{c'}$ and set for all other $\tilde{v} \in V$ the value of $y_{\tilde{v}}^{c',c}$ to the minimum of $u_{\tilde{v}}^{c',c}$ and $y_{c}^{c',c}$. This directly ensures that $y_{\tilde{v}}^{c',c}$ stays below $\tilde{x}_{\tilde{v}}^{c'}$, which together with the fact that $x_{\tilde{v}}^{c'}$ was nonzero when the first shift from $c'$ to $c$ occurs will help us bound the cost of the solution of the algorithm. Later we will show that this assignment to the centers also aligns with our connectivity constraints. Unfortunately the fact that $y_{\tilde{v}}^{c',\tilde{c}}$ can become $\tilde{x}_{\tilde{v}}^{c'}$ for multiple center $\tilde{c} \in C$ to which some opening of $c'$ gets shifted, will add an multiplicative factor of $O(k)$ to the cost of the solution (compared with $\tilde{x}$).

The fact that $y_{\tilde{v}}^{c',c}$ stays below $y_{c}^{c',c}$ and that this cannot cause a cut between a node $v' \in V$ and $S_c$ to get reduced below $y_{c}^{c',c}$ will be important to show that Property \eqref{half_inv_y} is preserved. The modification of $(x_{v'}^{c'})_{v' \in V}$ will be caused by adding $v$ to $D_{c'}$. To get a better intuition about the shifts Figure \ref{fig:example_shifts} depicts an example in which multiple shifts occur after each other.

One might note that after a shift from $c'$ over $v$ to $c$ that $x_v^{c'} = 0$ which together with the fact that the values of $x$ are only decreasing (which will be shown below) implies that we will only once shift opening from $c'$ over $v$ to a center $c$. Additionally after we have added $v$ to $C$ or a set $S_c$ there always exists a possible shift while $x_v^{c'} > 0$ for a $c' \in V$. Thus it is possible to execute the described algorithm for a polynomial time until for every $v,c' \in V$, $x_v^{c'} = 0$ and then return $\left(y_v^{c}\right)_{v \in V, c \in C}$. A more formal description can be found in Algorithm \ref{Alg_half}.

\begin{figure}
\begin{subfigure}{0.45\textwidth}
    \includegraphics[width=\textwidth, page =1]{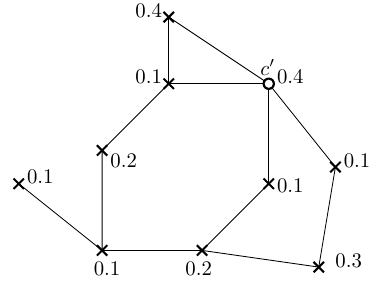}
    \subcaption{The initial values of $x_{v'}^{c'}$.}
    \label{fig:example_shofts_one}
\end{subfigure}
\hfill
\begin{subfigure}{0.45\textwidth}
    \includegraphics[width=\textwidth, page =2]{example_shifts_alt.pdf}
    \subcaption{The values $y_{v'}^{c',c_1}$ after a shift from $c'$ to $c_1$.}
\end{subfigure}
\hfill
\begin{subfigure}{0.45\textwidth}
    \includegraphics[width=\textwidth, page =3]{example_shifts_alt.pdf}
    \subcaption{The values $x_{v'}^{c'}$ after the first shift. }
    \label{fig:third}
\end{subfigure}
\hfill
\begin{subfigure}{0.45\textwidth}
    \includegraphics[width=\textwidth, page =4]{example_shifts_alt.pdf}
    \subcaption{The values $y_{v'}^{c',c_2}$ after a shift from $c'$ to $c_2$. }
\end{subfigure}
\begin{subfigure}{0.45\textwidth}
    \includegraphics[width=\textwidth, page =5]{example_shifts_alt.pdf}
    \subcaption{The values $x_{v'}^{c'}$ after the second shift.\newline \phantom{a}}
\end{subfigure}
\hfill
\begin{subfigure}{0.45\textwidth}
    \includegraphics[width=\textwidth, page =6]{example_shifts_alt.pdf}
    \subcaption{The values $y_{v'}^{c',c_1}$ after another shift from $c'$ to $c_1$ over $v \in S_{c_1}$.}
\end{subfigure}
        
\caption{An example in which multiple shifts from a node $c'$ to two centers $c_1$ and $c_2$ occur. The values are always given for all $v' \in V$. In the end the opening of $c'$ is $0$.}
\label{fig:example_shifts}
\end{figure}

\begin{algorithm}
	\DontPrintSemicolon
	\SetKwInOut{Input}{input}
	\SetKwInOut{Output}{output}
    \caption{Obtaining half-opened centers}\label{Alg_half}
    \textbf{Input:} $\tilde{x}$\;
	$r^* = 0$\;
    $C = \emptyset, Q = V$\;
    $\forall v,c,c' \in V: y_v^{c',c} = 0$\;
    $\forall v,c'  \in V: x_v^{c'} = \tilde{x}_v^{c'}$, $D_{c'} = \emptyset$\;
	\While {$\exists v,c: x_v^c \neq 0$}{
        $v^* = \argmin_{v \in Q} r_v$\;
        $r^* = r_{v*}$\;
        $Q = Q - v^*$\;
        \eIf{$\sum_{\tilde{c} \in V, d(\tilde{c},v^*) \leq 4r^*}x_{v^*}^{\tilde{c}} \geq \frac{1}{2}$}{
        $C = C \cup \{v^*\}$\;
        $S_{v^*} = \{v^*\}$\;
        }{
        $c^* = \argmin_{\tilde{c} \in C, y_{v^*}^{\tilde{c}} \geq \frac{1}{8k}} d(v^*,\tilde{c})$\;
        $S_{c^*} = S_{c^*} \cup \{v^*\}$
        }
        $r_{next} = \min_{v \in Q} r_v$\;
        \tcp{$r_{next} = \infty$ if $Q = \emptyset$}
        $c,c' = \argmin_{c_1 \in C, c_2 \in V, \sep^x(S_{c_1},c_2)> 0} d(c_1,c_2)$\label{line:choice_shift}\;
        \While{$d(c,c') \leq 4 r_{next}$ (and $c,c'$ are defined)}{
            \If{$y_c^{c',c} = 0$}{
                $\forall v \in V: u_v^{c',c} = x_v^{c'}$\;
            }
            \While{$\exists \tilde{v} \in S_c: x_{\tilde{v}}^{c'} > 0$}{
                $f = x_{\tilde{v}}^{c'}$ \label{line:begin_shift}\;
                $D_{c'} = D_{c'} \cup \{\tilde{v}\}$\;
                $\forall v \in V: x_v^{c'} = \Delta^{\tilde{x}}(D_{c'}, v, c')$\;
                $y_c^{c',c} = y_c^{c',c} + f$\;
                $\forall v \in V \setminus \{c\}: y_v^{c',c} = \min(u_v^{c',c}, y_{c}^{c',c})\label{line:end_shift}$\;
            }
            $c,c' = \argmin_{c_1 \in C, c_2 \in V, \sep^x(S_{c_1},c_2)> 0} d(c_1,c_2)$\;
        }
    }
    $\forall v \in V, c \in C: y_v^c = \sum_{c' \in V} y_v^{c',c}$\;
    \textbf{return} $(y_v^c)_{v,c\in V}$;
\end{algorithm}

To prove the correctness and analyze the approximation guarantee of the algorithm we proceed as follows:
In Lemma \ref{lem_half} we show that when the algorithm terminates, all centers are half-opened.
Then we show in Lemma \ref{lem:half_inv_number} the preservation of Property (\ref{half_inv_number}),
in Lemma \ref{lem:half_inv_x} for Property (\ref{half_inv_x}),
Lemma \ref{lem:y-shifts-lower-bounded-by-cut}, Corollary \ref{cor:v_con_S_c} and Lemma \ref{lem:half_inv_y} for Property (\ref{half_inv_y}) and after that we show the approximation guarantee in Lemma \ref{lem:half_inv_assigned}, Lemma \ref{lem:half_bound} and Theorem \ref{thm:bound_y}.
To begin, we will need the following technical properties of cuts. First of all we may observe that a cut between a set $S$ and a set $T$ is also a cut between $T$ and $S$. This ensures us that we actually can revert the flows with our shifts:

\begin{observation}
\label{obs:cut_sym}
For any $S, T \subseteq V$ any cut $N$ between $S$ and $T$ is also a cut between $T$ and $S$. Thus for any weight function $w$ it holds that $\sep^w(S,T) = \sep^w(T,S)$.
\end{observation}

Additionally we need that for any $S \subseteq V$ and any $t \in V$ the value $\Delta^w(S,v,t)$ only decreases if we add further elements to $S$ (thus ensuring that the values of $x$ decrease). Also a more generalized version of the statement that if $N$ is a cut between $S$ and $t$ then any cut between $N$ and $t$ is also a cut between $S$ and $t$ will be needed. Both lemmas will be proven in Section \ref{sec:technical}:

\begin{lemma}
\label{lem:decrease}
For all $S,S' \subseteq V$, $t,v \in V$ with $S \subseteq S'$ it holds that:
\begin{equation*}
 \Delta^w(S',v,t) \leq \Delta^w(S,v,t)
\end{equation*}
\end{lemma}

\begin{lemma}
\label{lem:cut_of_cut}
For any $S, S' \subseteq V$, $t \in V$ and any cut $N$ between $S$ and $t$ it holds that:
\begin{equation*}
\sep^w(N \cup S',t) \geq \sep^w(S \cup S',t)
\end{equation*}
\end{lemma}

Using these statements we will prove the correctness of the algorithm. First we may verify that actually each center $c$ in $C$ gets opened by at least one half. 

\begin{lemma}
\label{lem_half}
At the end of Algorithm \ref{Alg_half} for all $c \in C$ it holds that $y_c^c \geq \frac{1}{2}$.
\end{lemma}

\begin{proof}
    For any $c \in C$ we consider the iteration in which $c$ gets added to $C$. At the begin of this iteration the total amount of good assignments of $c$ is still at least $\frac{1}{2}$. Now for each of the corresponding nodes $c'$ there exists a possible shift from $c'$ to $c$ which would increase the opening of $c$ by $x_c^{c'}$. And since the distance between $c$ and $c'$ is at most $4r_v$ and for all other possible shifts going to another center $c_2 \in C$ the distance $d(c',c_2)$ is greater $4 r_c$ (as otherwise the shift would already have been performed in the previous iteration) we know that the algorithm will actually perform a shift from $c'$ to $c$ before performing any shifts going to another center. Thus at the end of the iteration in which $c$ gets added to $C$ we have that
    \begin{equation*}
        y_c^c = \sum_{c' \in V} y_c^{c',c} \geq \sum_{c' \in V, d(c',c) \leq 4r_c}x_{c} ^{c'} \geq \frac{1}{2}
    \end{equation*}
    where the values of $x$ are still referring to the situation at the begin of the iteration. Since the value of $y_c^c$ will not decrease afterwards this directly proves the lemma.
\end{proof}

For the remainder of the proof is important to note that the before and after any shift for any $v,c' \in V$ the value $x_v^{c'}$ is equal to $\Delta^{\tilde{x}}(D_{c'},v,c')$. This is obviously fulfilled at the begin of the algorithm as $\Delta^{\tilde{x}}(\emptyset,v,c') = \tilde{x}_v^{c'}$ and afterwards the value $x_v^{c'}$ gets updated accordingly when $D_{c'}$ gets modified.

\begin{observation}
\label{obs:value_of_x}
    Before any shift of the algorithm it holds for any $v,c' \in V$ that
    \begin{equation*}
        x_v^{c'} = \Delta^{\tilde{x}}(D_{c'},v,c').
    \end{equation*}
\end{observation}

Using this we can prove that during the entire algorithm the entirety of openings in $x$ and $y$ is at most $k$ and each node has a total assignment of at least one.

\begin{lemma}
\label{lem:half_inv_number}
For any $c$, $c'$ it holds that whenever a shift increases $y_c^{c',c}$ by a value $f$, $x_{c'}^{c'}$ gets decreased by $f$. Thus during the algorithm Property \eqref{half_inv_number} is preserved.
\end{lemma}

\begin{proof}
    We consider an arbitrary shift (lines \ref{line:begin_shift} - \ref{line:end_shift} of the algorithm) from $c'$ over $\tilde{v} \in S_c$ to $c$. During this shift $y_c^{c',c}$ gets increased by $f = x_{\tilde{v}}^{c'} = \Delta^{\tilde{x}}(D_{c'}, \tilde{v},c')$. At the same time the value of $x_{c'}^{c'}$ changes from $\Delta^{\tilde{x}}(D_{c'}, c',c')$ to $\Delta^{\tilde{x}}(D_{c'}\cup\{\tilde{v}\}, c',c')$. One might observe for any $S \subseteq V$ with $c' \in S$ that $c'$ is contained in any cut $N$ between $S$ and $c'$ while $\{c'\}$ is also a cut between $S$ and $c'$ which implies $\sep^{\tilde{x}}(S,c') = \tilde{x}_{c'}^{c'}$. Thus $\sep^{\tilde{x}}(D_{c'} \cup \{c'\},c') = \sep^{\tilde{x}}(D_{c'} \cup \{\tilde{v}\}\cup \{c'\},c') = \tilde{x}_{c'}^{c'}$. At the same time $\sep^{\tilde{x}}(D_{c'} \cup \{\tilde{v}\}, c') = \sep^{\tilde{x}}(D_{c'},c') + \Delta^{\tilde{x}}(D_{c'},\tilde{v},c') = \sep^{\tilde{x}}(D_{c'},c') + f$. Thus $\Delta^{\tilde{x}}(D_{c'}\cup\{\tilde{v}\}, c',c') = \Delta^{\tilde{x}}(D_{c'}, c',c') - f$ which means that the shift decreases $x_{c'}^{c'}$ by $f$ and thus the lemma holds. 
\end{proof}

\begin{lemma}
\label{lem:half_inv_assigned}
After any shift it holds for all $v \in V$ that:
\begin{equation*}
    \sum_{c \in C} y_{v}^{c',c} + x_v^{c'} \geq \tilde{x}_v^{c'}
\end{equation*}
\end{lemma}

\begin{proof}
The claim can be shown via induction. Obviously in the beginning the claim is true because $x_v^{c'} = \tilde{x}_v^{c'}$. Let us now consider a shift from $c'$ to the center $c$ over the node $\tilde{v} \in S_c$. Then adding $\tilde{v}$ to $D_{c'}$ increases $\sep^{\tilde{x}}(D_{c'},c')$ by exactly $\Delta^{\tilde{x}}(D_{c'}, \tilde{v},c') = x_{\tilde{v}}^{c'} = f$. At the same time for any $v \in V$ the value $\sep^{\tilde{x}}(D_{c'} \cup \{v\}, c')$ is not decreasing when we add $\tilde {v}$ to $D_{c'}$ since any cut between $D_{c'} \cup \{v\} \cup \{\tilde{v}\}$ and $c'$ (including in particular the minimum cut) is also a cut between $D_{c'} \cup \{v\}$ and $c'$. Thus $x_{v}^{c'} = \Delta^{\tilde{x}}(D_{c'}, v, c')$ decreases by at most $f$. At the same time $y_{c}^{c',c}$ gets increased by $f$ and $y_v^{c',c} $ gets set to $ \min(u_v^{c',c}, y_{c}^{c',c})$. So either it gets increased by $f$ and the induction holds or $y_v^{c',c} = u_v^{c',c}$. In the latter case we know by induction that in the moment in which $u_v^{c',c}$ was set to $x_v^{c'}$, $\sum_{\tilde{c} \in C} y_{v}^{c',\tilde{c}} + u_v^{c',c} = \sum_{\tilde{c} \in C} y_{v}^{c',\tilde{c}} + x_v^{c'} \geq \tilde{x}_v^{c'}$.
Since $y_{v}^{c',c}$ has in the intermediate iterations increased from $0$ to $u_v^{c',c}$ and for all $\tilde{c}$ the value of $y_{v}^{c',\tilde{c}}$ did not decrease we may conclude that the claim also remains true in this case.
\end{proof}

Furthermore we can verify that at any point during the algorithm any assignment by $x$ is upper bounded by the respective minimum cut:

\begin{lemma}
\label{lem:half_inv_x}
    After any shift of the algorithm it holds for any $v,c' \in V$ that
    \begin{equation*}
        \sep^x(v,c') \geq x_v^{c'}.
    \end{equation*}
    Thus Property \eqref{half_inv_x} is preserved.
\end{lemma}

\begin{proof}
Assume that the lemma is wrong. Then at some point during the algorithm there exist two nodes $v, c' \in V$ such that $\sep^x(v,c') < x_v^{c'}$ which means that there exists a cut $N$ between $v$ and $c'$ such that $\sum_{w \in N} x_{w}^{c'} < x_v^{c'}$. Let $N = \{w_1,\dots,w_\ell\}$ be such a cut. We will now bound $\sep^{\tilde{x}}(D_{c'} \cup N, c') $ as follows:
\begin{align*}
    \sep^{\tilde{x}}(D_{c'} \cup N, c') &\stackrel{\phantom{\text{Lem. }\ref{lem:decrease}}}{=} \sep^{\tilde{x}}(D_{c'},c') + \sum_{i = 1}^\ell \big( \sep^{\tilde{x}}(D_{c'} \cup\{w_1,\dots,w_i\},c') \\
     &\stackrel{\phantom{\text{Lem. }\ref{lem:decrease}}}{} \phantom{\sep^{\tilde{x}}(D_{c'},c') + \sum_{i = 1}^\ell \big(} -\sep^{\tilde{x}}(D_{c'} \cup\{w_1,\dots,w_{i -1}\},c') \big)\\
    &\stackrel{\phantom{\text{Lem. }\ref{lem:decrease}}}{=} \sep^{\tilde{x}}(D_{c'},c') + \sum_{i=1}^\ell \Delta^{\tilde{x}}(D_{c'} \cup \{w_1,\dots,w_{i-1}\},w_i,c')\\
    &\stackrel{\text{Lem. }\ref{lem:decrease}}{\leq} \sep^{\tilde{x}}(D_{c'},c') + \sum_{i=1}^\ell\Delta^{\tilde{x}}(D_{c'},w_i,c')\\
    &\stackrel{\text{Obs. }\ref{obs:value_of_x}}{=} \sep^{\tilde{x}}(D_{c'},c') + \sum_{i=1}^\ell x_{w_i}^{c'}\\
    &\stackrel{\phantom{\text{Lem. }\ref{lem:decrease}}}{<} \sep^{\tilde{x}}(D_{c'},c') + x_v^{c'}\\
    &\stackrel{\text{Obs. }\ref{obs:value_of_x}}{=} \sep^{\tilde{x}}(D_{c'} \cup\{v\}, c')
\end{align*}
However since $N$ is a cut between $v$ and $c'$ this directly contradicts Lemma \ref{lem:cut_of_cut}. Thus the lemma is correct.
\end{proof}

Using this we may prove that for any $v \in V$, $c \in C$ there exists a flow with capacity $y_v^c$ from $v$ to $c$ which together with the fact that $x_v^c \geq \frac{1}{8k}$ for any $v \in S_c$ can be used to show that Property \eqref{half_inv_y} is preserved:

\begin{lemma}\label{lem:y-shifts-lower-bounded-by-cut}
After every shift of the algorithm it holds for any $v \in V, c \in C, c' \in V$ and all cuts $N$ between $v$ and $S_c$
\begin{equation*}
    y_v^{c',c} \leq \sum_{w \in N} y_w^{c',c}.
\end{equation*}
\end{lemma}

\begin{proof}
First we consider the situation that $v = c'$. Assume that after the $\ell$-th shift the claim is violated. Let $v_1,\dots,v_\ell$ be the nodes over which opening has been shifted from $c'$ to $c$, $D_1,\dots,D_\ell$ be the respective states of $D_{c'}$ before and $f_1,\dots,f_\ell$ the values of $f$ during these shifts. One might observe that $\forall w \in V: u_w^{c',c} = \Delta^{\tilde{x}}(D_1,w,c')$ and that after the $\ell$-th shift $v_1,\dots,v_\ell \in S_c$. Let $N = \{w_1,\dots,w_m\}$ be an arbitrary cut between  $c'$ and $S_c$. Then by Observation \ref{obs:cut_sym}, $N$ is also a cut between $S_c$ and $c'$ and it holds that:
\begin{align*}
    \sum_{i = 1}^m  u_{w_i}^{c',c} &\stackrel{\phantom{\text{Lem. }\ref{lem:decrease}}}{=} \sum_{i= 1}^m \Delta^{\tilde{x}}(D_1,w_i,c')\\
    &\stackrel{\text{Lem. }\ref{lem:decrease}}{\geq} \sum_{i = 1}^\ell \Delta^{\tilde{x}}(D_1\cup\{w_1,\dots,w_{i-1}\},w_i,c')\\
    &\stackrel{\phantom{\text{Lem. }\ref{lem:decrease}}}{=} \sep^{\tilde{x}}(D_1 \cup N,c') - \sep^{\tilde{x}}(D_1,c')\\
    &\stackrel{\text{Lem. }\ref{lem:cut_of_cut}}{\geq} \sep^{\tilde{x}}(D_1 \cup S_c,c')- \sep^{\tilde{x}}(D_1,c')
\end{align*}

But at the same time
\begin{align*}
y_{c'}^{c',c} &\stackrel{\phantom{\text{Lem. }\ref{lem:decrease}}}{=} \sum_{i=1}^\ell f_i\\
&\stackrel{\phantom{\text{Lem. }\ref{lem:decrease}}}{=} \sum_{i = 1}^\ell \Delta^{\tilde{x}}(D_i,v_i,c')\\
&\stackrel{\text{Lem. } \ref{lem:decrease}}{\leq} \sum_{i=1}^\ell \Delta^{\tilde{x}}(D_1 \cup\{v_1,\dots,v_{i-1}\},v_i,c')\\
&\stackrel{\phantom{\text{Lem. }\ref{lem:decrease}}}{=} \sep^{\tilde{x}}(D_1 \cup\{v_1,\dots,v_\ell\},c') - \sep^{\tilde{x}}(D_1,c')\\
&\stackrel{\phantom{\text{Lem. }\ref{lem:decrease}}}{=} \sep^{\tilde{x}}(D_1 \cup S_c,c')- \sep^{\tilde{x}}(D_1,c').
\end{align*}

Thus $y_v^{c',c} \leq \sum_{w \in N} u_w^{c',c}$ and since for all $w \in V$ it holds after the $\ell$-th shift that $y_w^{c',c} = \min(y_c^{c',c}, u_w^{c',c})$ we may conclude that $y_v^{c',c} \leq \sum_{w \in N} y_w^{c',c}$ which means that our assumption was wrong. 

Now we consider the case that $v \neq c'$. Let $N$ be an arbitrary cut between $v$ and $S_c$. If $c' \in H_{S_c}(N)$ then by the previous argument $\sum_{w \in N} y_w^{c',c} \geq y_{c'}^{c',c} \geq y_v^{c',c}$. Otherwise Lemma \ref{lem:half_inv_x} tells us that Property \eqref{half_inv_x} was fulfilled when $(u^{c',c}_w)_{w \in V}$ was set. Since $N$ also separates $v$ and $c'$ in this case we may conclude that $u_v^{c',c} \leq \sum_{w \in N} u_w^{c',c}$ and since for all $w$ it holds that $y_w^{c',c} = \min\left(y_{c}^{c',c}, u_w^{c',c}\right)$ we also get $y_v^{c',c} \leq \sum_{w \in N} y_w^{c',c}$. Thus in both cases our claim is true and the lemma holds.
\end{proof}

\begin{corollary}
\label{cor:v_con_S_c}
After every shift of the algorithm it holds for any $v \in V, c \in C$ and every cut $N$ between $v$ and $S_c$ that:
\begin{equation*}
    y_v^{c} \leq \sum_{v' \in N} y_{v'}^c
\end{equation*}
\end{corollary}

\begin{lemma}
\label{lem:half_inv_y}
At any point during the algorithm for any $c \in C$ and any $v \in S_c \setminus \{c\}$ it holds that $\sep^y(v,c)\geq \frac{1}{8k}$. Thus at any point of the algorithm Property \eqref{half_inv_y} is preserved.
\end{lemma}

\begin{proof}
We will show the claim via induction. In the beginning it is obviously fulfilled. One may note that since for all $\tilde{c} \in C,v' \in V$ the value of $y_{v'}^{\tilde{c}}$ is only increasing, the claim can only be violated by adding a new node $v$ to a set $S_c$, i.e. setting $S_c = S_c \cup \{v\}$. By construction of the algorithm we know that before this $y_v^c \geq \frac{1}{8k}$ which together with Corollary \ref{cor:v_con_S_c} implies that the weight of the minimum cut between $v$ and $S_c$ is at least $\frac{1}{8k}$. Let now $N$ be an arbitrary cut between $v$ and $c$. If $H_c(N)$ does not contain any node from $S_c$ then it is also a cut between $v$ and $S_c$ and thus has weight at least $\frac{1}{8k}$. If there exists however a node $w \in S_c$ with $w \in H_c(N)$ then the weight of the cut is at least $\sep^{y}(w,c)$ which by induction is at least $\frac{1}{8k}$. Thus $\sep^{y}(v,c) \geq \frac{1}{8k}$ and the induction holds.
\end{proof}

Finally we will show that the total assignment of any $v \in V$ by $y$ is not too much larger than the total assignment by $\tilde{x}$ and limit the distance of $v$ to any center it gets assigned to:

\begin{figure}
\centering
    \includegraphics[width= 0.6 \textwidth]{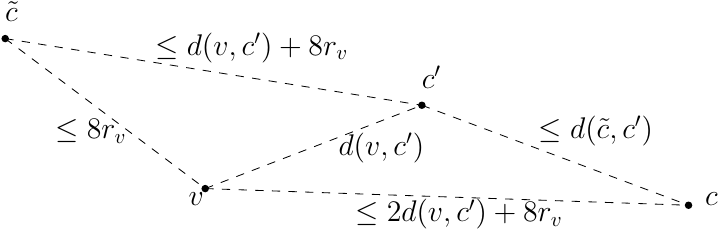}
    \caption{A sketch depicting the third case in Lemma \ref{lem:half_bound}. }
    \label{fig:distance_formula}
\end{figure}

\begin{lemma}
\label{lem:half_bound}
At any point during the algorithm for all $c \in C, c', v \in V$ it holds that:
\begin{enumerate}[(i)]
        \item If $v \neq c$ then $y_v^{c',c} \leq \tilde{x}_v^{c'}$
        \item $d(v,c) \leq 2 d(v,c') + 8 r_v$ if $ y_v^{c',c} > 0$
\end{enumerate}
\end{lemma}

\begin{proof}
First we prove (i). One might observe that $y_v^{c',c}$ is upper bounded by $u_v^{c',c}$ which is initially set to 
\begin{equation*}
\Delta^{\tilde{x}}(D_{c'},v,c') \stackrel{\text{Lem. }\ref{lem:decrease}}{\leq} \Delta^{\tilde{x}}(\emptyset,v,c') = \sep^{\tilde{x}}(v,c')
\end{equation*}
Thus (i) is fulfilled.

To prove (ii) let $y_v^{c',c} > 0$ which implies that $u_v^{c',c} > 0$ (unless $v=c$ but in this case the property is trivially fulfilled). Again we consider the moment in which  $u_v^{c',c}$ is set. We consider three cases of which exactly one is fulfilled in that moment:
\begin{itemize}
\item $v \in Q$: In this case we know that $r_v \geq r_{next}$ and since mass gets shifted from $c'$ to $c$ we know that $d(c,c') \leq 4r_{next}$.
Thus by the triangle inequality $d(v,c) \leq d(c,c') + d(v,c') \leq d(v,c')+4r_v$.
\item $v \in C$: Since $u_v^{c',c} > 0$ we know that $x_v^{c'}> 0$ which means that $v$ and $c'$ have been considered in line \ref{line:choice_shift} of Algorithm \ref{Alg_half} which together with the fact that $c,c'$ have been chosen for the shift implies that $d(v,c') \geq d(c,c')$ which together with the triangle inequality gives us $d(v,c) \leq 2d(v,c')$.
\item There exists a $\tilde{c} \in C$ such that $v \in S_{\tilde{c}}$: By the same arguments as in the previous step we get that $x_v^{c'} > 0$ which implies $\sep^x(S_{\tilde{c}},c') > 0$ which implies that $d(\tilde{c},c') \geq d(c,c')$.
 At the same time we know that in the iteration in which $v$ was added to $S_{\tilde{c}}$ at least $\frac{1}{4}$ of good assignments, which as a reminder were assignments to nodes $z \in V$ with $d(v,z) \leq 4r_v$, have already been shifted and for any previous shift from $z \in V$ to $c'' \in C$ that $d(z,c'') \leq 4r_v$. Thus there exists at least one $c'' \in C$ to which $\frac{1}{8k}$ of the good assignments have been shifted which means $d(c'',v) \leq 8r_v$ and $y_v^{c''} \geq \frac{1}{8k}$.
Thus $d(v, \tilde{c}) \leq 8 r_v$ which implies $d(\tilde{c},c') \leq d(v,c') + 8r_v$ and thus $d(v,c) \leq 2 d(v,c') + 8 r_v$. A visualization of this situation can be found in Figure \ref{fig:distance_formula}.
\end{itemize}
Since in any case $d(v,c) \leq 2 d(v,c') + 8 r_v$ we can conclude that this is true for any $v,c',c$ with $y_v^{c',c} > 0$ which directly proves (ii).

\end{proof}

By combining all these results we obtain:

\begin{restatable}{theorem}{HalfOpened}
\label{thm:bound_y}
Algorithm \ref{Alg_half} computes an $\frac{1}{8k}$-connected solution $y$ such that every center is at least half opened and $\cost(y) \leq 20k \cost(\tilde{x})$.
\end{restatable}

\begin{proof}
It is easy to verify that the algorithm is terminating and that at the end for all $v,c \in V$ it holds that $x_v^c = 0$. Thus we obtain a solution which because of Lemma \ref{lem:half_inv_assigned}, \ref{lem:half_inv_number} and \ref{lem:half_inv_y} fulfills the following properties:
\begin{align*}
         \sum_{c \in V} y_v^c & \geq 1 & \forall v \in V\\
         \sum_{c \in V} y_c^c&\leq k& \\
         \sum_{v' \in S} y_{v'}^c &\geq \min\left(\frac{1}{8k}, y_{v}^c\right) & \forall v,c \in V, S \in \mathcal{S}_{v,c}
\end{align*}
Thus the algorithm actually returns an $\frac{1}{8k}$-connected solution which by Lemma \ref{lem_half} contains only half opened centers. To bound the cost of this solution we consider an arbitrary node $v \in V$ its cost in $y$ are bounded by:
\begin{align*}
\sum_{c \in C} y_v^c d(v,c) &\stackrel{\phantom{\text{Lem. }\ref{lem:half_bound}}}{=} \sum_{c' \in V}\sum_{c \in C} y_v^{c',c} d(v,c)\\
&\stackrel{\text{Lem. }\ref{lem:half_bound}}{\leq} \sum_{c' \in V} \sum_{c \in C: d(v,c) \leq 2d(v,c') + 8r_v} \tilde{x}_v^{c'} d(v,c)\\
&\stackrel{\phantom{\text{Lem. }\ref{lem:half_bound}}}{\leq} \sum_{c' \in V} 2k \cdot \tilde{x}_v^{c'} (2d(v,c') + 8 r_v)\\
&\stackrel{\phantom{\text{Lem. }\ref{lem:half_bound}}}{=} 2k \left( 2\sum_{c' \in V}\tilde{x}_v^{c'}d(v,c') + 8r_v\sum_{c' \in V}\tilde{x}_v^{c'}\right)\\
&\stackrel{\phantom{\text{Lem. }\ref{lem:half_bound}}}{=} 2k\left(10p_v\right) = 20k \cdot p_v
\end{align*}
 Note that Lemma \ref{lem:half_bound} does not hold for values $y_v^{c',c}$ with $v=c$ but since the summand $y_v^{c',c}d(v,c)$ equals $0$ in this case, this does not make a difference.
Thus for every node the cost increases by at most a factor of $20k$ which means that $\cost(y) \leq 20k \cost(\tilde{x})$.

\end{proof}

\subsection{Finding suitable replacements.}\label{sec:replace}

In this section we want to subdivide $C$ into two subsets $C_1$ and $C_{1/2}$. Intuitively we can consider $C_1$ as the centers opened with value $1$ while the centers in $C_{1/2}$ can be considered to be opened with value $\frac{1}{2}$ (again following the ideas of the LP rounding algorithm by Charikar et al. \cite{charikar1999constant}). Following this logic we require that $|C_1| + \frac{1}{2} |C_{1/2}| \leq k$. However it is important to note that we are not actually modifying the openings of the centers yet since this would require to also raise some of the assignments to the centers to ensure a sufficient connectivity which would cause some problems in the next step.

For each center $c$ in $C_{1/2}$ we will determine a successor $s(c) \in C$ such that reassigning all nodes (partially) assigned to $c$ from $c$ to $s(c)$ will not produce too much additional costs. Later we will then open at most half the centers in $C_{1/2}$ such that for every center either itself or its successor gets opened, thus obtaining at most $k$ centers such that every assignment can be redirected to one of them. For simplicity we will for every center $c \in C$ define $M_c$ to be the total assignment (or mass) assigned to it, i.e., $M_c = \sum_{v \in V} y_v^c$. Then replacing $c$ by $s(c)$ will cost up to $M_c \cdot d(c,s(c))$ since $s(c)$ might be that much farther away from the respective nodes plus any cost necessary to increase some previous assignments to $s(c)$ to ensure a sufficient connectivity. Let $R(c,\tilde{c})$ denote the cost that arises if we replace a center $c$ by another center $\tilde{c}$.  We will provide upper bounds for $R(c,s(c))$ for any $c \in C_{1/2}$.

First let us formally define how a center $c \in C$ can be replaced by another center $\tilde{c} \in C$ for which $y_c^{\tilde{c}} > 0$. During this procedure we will further relax our connectivity constraint and only require the solution to be $\frac{1}{16k}$-connected afterwards. We distinguish two cases: 
\begin{itemize}
    \item If $y_{c}^{\tilde{c}} \geq \frac{1}{16k}$ one can set for all $v \in V \setminus\{\tilde{c}\}$ the value $y_v^{\tilde{c}}$ to $\min(1,y_v^{c}+y_v^{\tilde{c}})$ and reduce $y_v^{c}$ to $0$. Lemma \ref{lem:half_inv_y} together with the fact that $y_{c}^{\tilde{c}}$ was at least $\frac{1}{16k}$ before this replacement ensures that the instance stays $\frac{1}{16k}$-connected. It holds that
    \begin{equation}
        R(c,\tilde{c}) \leq \sum_{v\in V} d(v,\tilde{c}) - d(v,c) \leq M_c \cdot d(c,\tilde{c}).
    \end{equation}
\item If $y_{c}^{\tilde{c}} < \frac{1}{16k}$ we need to be more careful to guarantee a sufficient connectivity. For simplicity reasons we will ensure this connectivity by simply multiplying all assignments to $\tilde{c}$ by $\frac{1}{y_c^{\tilde{c}}}\frac{1}{16k}$. This directly increases $\sep^y(c,\tilde{c})$ from $y_{c}^{\tilde{c}}$ to $\frac{1}{16k}$. One might note that in practice there would be cheaper methods to do this, for example only increasing the assignments of the nodes participating to the maximum flow from $c$ to $\tilde{c}$. The additional cost of this step are upper bounded by
\begin{equation*}
    \sum_{v \in V} y_v^{\tilde{c}} \cdot d(v, \tilde{c})\frac{1}{y_c^{\tilde{c}}}  \frac{1}{16k}.
\end{equation*}
Afterwards the same steps as in the previous case can be applied leading to an entire replacement cost of
\begin{equation}
    R(c,\tilde{c}) \leq M_c \cdot d(c,\tilde{c}) + \sum_{v \in V} y_v^{\tilde{c}} \cdot d(v, \tilde{c})\frac{1}{y_c^{\tilde{c}}}  \frac{1}{16k}.\label{eq:replacement_hard}
\end{equation}
\end{itemize}

In both cases we need to bound the value $M_c \cdot d(c,\tilde{c})$ if we want to replace $c$ by $\tilde{c}$. To do this, we first prove for any point $v$ at least partially assigned to $c$ that the initial radius of $c$ is not too large compared to $v$'s own radius $r_v$ and its initial assignments:

\begin{lemma}
\label{lem:radius_shift}
At any point during Algorithm \ref{Alg_half} if $y_v^{c',c} > 0$ then $r_c \leq 0.25 d(v,c') + 2 r_v$.
\end{lemma}

\begin{proof}
Let  $v,c' \in V$ be two nodes and  $c \in C$ a center such that $y_v^{c',c}>0$. Then if $v \neq c$ also $u_v^{c',c}>0$ (for $v=c$ the claim would be trivially fulfilled).
Let us consider the moment in which $u_v^{c',c}$ is set. Then exactly one of the following cases is fulfilled:
\begin{itemize}
\item $v \in Q$: Then $r_v \geq r^*$ while $r^* \geq r_c$ since $c \in C$. Thus $r_c \leq r_v$.
\item $v \in C$: Since $u_v^{c',c} > 0$ we know that $x_v^{c'} > 0$.
Assume for the sake of contradiction that $d(v,c') \leq 4r^*$.
If $v$ was added to $C$ in an earlier iteration, then we would already have applied a shift from $c'$ to $v$ in the previous iteration, as the algorithm applies all possible shifts with distance at most $4r_{next}$, which is always equal to the next value of $4r^*$.
Now assume that $v$ was added in the current iteration. By the fact, that $u_v^{c',c}$ is was not set earlier we know that until now no shift from $c'$ to $c$ was executed. By the same argument as before, we get that $d(c',c) > 4 r^*$.
Thus we would have applied a shift from $c'$ to $v$ before applying a shift from $c'$ to $c$.

In both cases we would already have applied a shift from $c'$ to $v$ contradicting that $x_v^{c'} > 0$. This implies that $d(v,c') > 4r^*$ and by combining this with the fact that $r_c \leq r^*$ we obtain that $r_c \leq 0.25 d(v,c')$.
\item There exists a $\tilde{c} \in C$ such that $v \in S_{\tilde{c}}$: As in the proof of Lemma \ref{lem:half_bound} we may conclude that $d(v,\tilde{c}) \leq 8 r_v$. Thus as in the previous case we would already have applied a shift from $c'$ over $v$ to $\tilde{c}$ if $d(v,c') + 8 r_v \leq 4 r^*$ and hence $r_c \leq r^* \leq 0.25 d(v,c') + 2 r_v$.
\end{itemize}

\end{proof}

Using this we may conclude that if we reassign each center $c$ to a center at distance at most $O(r_c)$ then this will also increase the costs of the assignments only by $O(k \cdot \cost(\tilde{x}))$:

\begin{lemma}
\label{lem:shift_by_radius}
\begin{equation*}
    \sum_{c \in C} \sum_{v \in V} y_v^c r_c \leq 4.5 k \cdot \cost(\tilde{x})
\end{equation*}
\end{lemma}
    
\begin{proof}

We will consider the contribution of all assignments of a single node $v$ to the sum:
\begin{align*}
    \sum_{c \in C} y_v^c r_c &\stackrel{\phantom{\text{Lem. }\ref{lem:decrease}}}{=} \sum_{c \in C} \sum_{c' \in V} y_v^{c',c} r_c\\
    &\stackrel{\text{Lem. }\ref{lem:radius_shift}}{\leq} \sum_{c \in C} \sum_{c' \in V} y_v^{c',c} \left( 0.25 d(v,c') + 2 r_v \right)\\
    &\stackrel{\text{Lem. }\ref{lem:half_bound}}{\leq} \sum_{c  \in C} \sum_{c' \in V} \tilde{x}_v^{c'} \left( 0.25 d(v,c') + 2 r_v \right)\\
    &\stackrel{\phantom{\text{Lem. }\ref{lem:decrease}}}{\leq} 2k \left( \sum_{c' \in V} \tilde{x}_v^{c'} 0.25 d(v,c') + \sum_{c' \in V} \tilde{x}_v^{c'} 2 r_v \right)\\
    &\stackrel{\phantom{\text{Lem. }\ref{lem:decrease}}}{=} 4.5k \cdot p_v
\end{align*}

Since $p_v$ is the cost of $v$ in $\tilde{x}$ we obtain that:
\begin{equation*}
    \sum_{c \in C} \sum_{v \in V} y_v^c r_c \leq \sum_{v \in V} 4.5k \cdot p_v = 4.5k \cdot \cost(\tilde{x})
\end{equation*}

which proves the lemma.
\end{proof}

Now for each center $c$ we define $a_c = 1 - y_c^c$ as the amount by which the center is unopened. One might note that there exists at least $a_c$ assignments of $c$ to other centers $\tilde{c} \in C$. By considering only the first two cases in the proof of Lemma \ref{lem:half_bound} we can bound the length of these assignment as follows:

\begin{corollary}
\label{cor:center_shift}
For any centers $c, \tilde{c} \in C$ and any point $c' \in V$ it holds that if $y_c^{c',\tilde{c}} > 0$ that:
\begin{equation*}
    d(c,\tilde{c}) \leq d(c,c') + \max\left(4 r_c,d(c,c')\right)
\end{equation*}
\end{corollary}

We will now subdivide the center set into two sets: The set $H = \{c \in C \mid a_c > \frac{1}{4}\}$ of centers that are closer to being half-opened than opened and the set $O = \{c \in C \mid a_c \leq \frac{1}{4}\}$ of centers that are closer to being opened. It is important to note that despite some similarities there is no direct connection to the division of $C$ into $C_{1/2}$ and $C_1$ that we want to obtain. While $H$ will in its entirety be added to $C_{1/2}$ also some nodes in $O$ might be added to it.

For the centers $c$ in $H$ we will observe that there will always exist another center $\tilde{c}$ such that $\tilde{c}$ will not be too far away from $c$ and there already exists a $\frac{1}{16k}$ assignment from $c$ to $\tilde{c}$:

\begin{lemma}\label{lem:H-connectivity-and-cost-guarantee}
For any $c \in H$ there exists a center $\tilde{c} \in C \setminus \{c\}$ with $y_c^{\tilde{c}}\geq \frac{1}{16k}$ and $d(c,\tilde{c}) \leq 16 r_c$.
\end{lemma}

\begin{proof}
By Markov's inequality it holds that initially at least $\frac{7}{8}$ of $c$ has been assigned to openings within distance $8 r_c$. Since $y_c^c < \frac{3}{4}$ we thus get:
\begin{alignat*}{2}
    &\sum\limits_{c' \in V, d(c,c')\leq 8 r_c} \sum\limits_{\tilde{c} \in C \setminus \{c\}} y_{c}^{c',\tilde{c}} &&\quad> \frac{7}{8} - \frac{3}{4} = \frac{1}{8}\\
    \stackrel{\text{Cor. }\ref{cor:center_shift}}{\Rightarrow}\quad &\sum\limits_{\tilde{c} \in C \setminus \{c\}: d(c,\tilde{c}) \leq 16 r_c} y_c^{\tilde{c}} &&\quad> \frac{1}{8}
\end{alignat*}

By combining this with the fact that there exist at most $2k$ centers we may use Markov's inequality again to conclude that there exists at least one of them such that $y_c^{\tilde{c}}\geq \frac{1}{16k}$ and $d(c,\tilde{c}) \leq 16 r_c$.
\end{proof}
Using this lemma we may conclude for any $c \in H$ that if we set 
\begin{equation*}
s(c) = \argmin_{\tilde{c} \in C\setminus \{c\}, x_c^{\tilde{c}} \geq \frac{1}{16k} } d(c,\tilde{c}),
\end{equation*}
then $R(c,s(c)) \leq 16r_c M_c$. 

Now for the centers in $O$ we would like to find a similar successor. However we cannot always require that the minimum cut between a center $c$ and a possible successor $\tilde{c}$ is already greater than $\frac{1}{16k}$. Thus we will need to bound the replacement using \ref{eq:replacement_hard}.

%
%

The problem with this bound is that if for a particular $c$ the value of $a_c$ is rather small then $c$ gets only assigned to other centers by a really small portion which can result in very large values of $\sum_{v \in V} y_v^{\tilde{c}} \cdot d(v, \tilde{c})\frac{1}{y_c^{\tilde{c}}}  \frac{1}{16k}$. Additionally these centers can also be quite far away if they resulted from assignments of $c$ in $\tilde{x}$ whose distance was way higher than the radius $r_c$. However, if all the centers in $O$ have very small $a_c$-values then $|O|$ is not much larger then $\sum_{c \in O} y_c^c$ which means that we would only need to add few nodes in $O$ to $C_{1/2}$ to add the other centers to $C_1$ which results in them not producing any additional cost at all. Thus we would have few centers with high costs and many centers with very low additional cost. To split the cost more evenly between the centers in $O$ we define for each $c \in O$ a potential $\Phi_c$:

\begin{equation*}
    \Phi_c = 3 M_c \cdot r_c + \frac{1}{16k} \cost(y)
\end{equation*}

As we will show later, the sum of all of these potentials lies only in $O(k\log n) \cdot \cost(\tilde{x})$. Additionally we will show that for any $c \in O$, the cost $R(c,\tilde{c})$ of the best reassignment to any other center $\tilde{c} \in C$ is upper bounded by $\frac{\Phi_c}{a_c}$. Using this we can bound the entire additional cost incurred by adding centers from $O$ to $C_{1/2}$ by $2$ times the sum of all potentials if we always add the center to $C_{1/2}$ that minimizes $\frac{\Phi_c}{a_c}$ among all remaining centers until $|C \setminus C_{1/2}| + \frac{1}{2}|C_{1/2}| \leq k$. Once the latter inequality is fulfilled we can simply add all the remaining centers to $C_1$.

However there is one more detail regarding the decision whether or not a center in $O$ gets added to $C_{1/2}$. As mentioned above, half of the centers contained in $C_{1/2}$ get also opened in their entirety while the others get closed in a later step. To do this we are going to consider the graph using the centers as nodes and having a directed edge from each center in $C_{1/2}$ to its successor. If this graph is bipartite we can simply open the smaller half of the bipartition. However initially the graph might still contain cycles of uneven length. So instead of directly requiring a bipartite graph we will settle for the property that any cycle only contains at most one node from $O$. Using this property we will later be able to modify the graph such that the cycles of uneven length get removed. Every time a node in $O$ gets added to $C_{1/2}$, one can follow the directed path of edges between centers and their successors until either a previous node gets encountered (i.e., there is a cycle) or the path ends. If the path ends one can check whether the final node $r$ was neither added to $C_{1/2}$ nor $C_1$ yet. If yes, it gets added to $C_1$ which will avoid cycles containing multiple nodes in $O$ since the centers in $C_1$ have no successors. Algorithm \ref{Alg_h_integral} gives a more formal description how the sets $C_1$ and $C_{1/2}$ can be obtained.

\begin{algorithm}
	\DontPrintSemicolon
	\SetKwInOut{Input}{input}
	\SetKwInOut{Output}{output}
    \caption{Splitting up the centers into $C_1$ and $C_{1/2}$}\label{Alg_h_integral}
	$H = \left\{c \in C \mid y_c^c < \frac{3}{4}\right\}$\;
    $O = C \setminus H$\;
    $\forall c \in C: a_c = 1 - y_c^c$\;
    $C_{1},C_{1/2} = \emptyset$\;
    \For{$c \in H$}{
        Set $s(c) = \argmin_{\tilde{c} \in C \setminus \{c\}: y_c^{\tilde{c}} \geq \frac{1}{16k}} d(c,\tilde{c})$\;
        $C_{1/2} = C_{1/2} \cup \{c\}$
    }
    Set $O' = O$\;
	\While {$|O'| + |C_1| + \frac{1}{2} |C_{1/2}| > k$}{\label{line:while}
        $c = \argmin_{\tilde{c} \in O'} \frac{\Phi_{\tilde{c}}}{a_{\tilde{c}}}$\;
        Set $s(c) = \argmin_{\tilde{c} \in C \setminus \{c\}} R(c, \tilde{c})$\;
        Add $c$ to $C_{1/2}$, remove it from $O'$\;
        Set $S = \left\{(\tilde{c},s(\tilde{c})) \mid \tilde{c} \in C_{1/2}\right\}$, $G_S=(C,S)$\;
        Check whether the single directed path from $c$ in $G_S$ ends in a node $r$ or in a cycle\;
        \label{line:if} \If{$r$ exists and $r \in O'$}{
            Add $r$ to $C_1$, remove it from $O'$\;
        }
    }
    $C_1 = C_1 \cup O'$\;
    \textbf{return} $C_1,C_{1/2}, s$\;
\end{algorithm}

To show the correctness of the algorithm, it is first important to prove that the algorithm actually terminates. This can be done by showing that the opening of the centers neither in $C_1$ nor $C_{1/2}$ plus the size of $C_1$ and half the size of $C_{1/2}$ is at most $k$. Since the algorithm removes in every step of the while loop at least one center from $O'$ this directly implies that the while loop only runs for at most $|O|$ iterations:

\begin{lemma}
\label{lem:inv_half_integral}
Whenever line \ref{line:while} of Algorithm \ref{Alg_h_integral} is reached it holds that
\begin{equation*}
    |C_1| + \frac{1}{2}|C_{1/2}| + \sum_{c \in O'} y_c^c \leq k.
\end{equation*}
\end{lemma}

\begin{proof}
We will show the claim via induction. The first time the condition of the while loop gets checked, it holds that
\begin{align*}
    |C_1| + \frac{1}{2}|C_{1/2}| + \sum_{c \in O'} y_c^c &= |\emptyset| + \sum_{c\in H} \frac{1}{2} + \sum_{c \in O} y_c^c\\
    &\leq \sum_{c\in H} y_c^c + \sum_{c \in O} y_c^c\\
    &= \sum_{c \in C} y_c^c\\
    &\leq k.
\end{align*}

Let us now assume that the claim holds before the $t$-th iteration of the loop. Let $c_t$ be the center considered during the $t$-th loop. If we only remove $c_t$ from $O'$ and add it to $C_{1/2}$ it is clear that this only decreases the sum since $y_c^c \geq \frac{3}{4}$ which means that the claim remains true. If we also remove a node $r$ from $O'$ and add it to $C_1$ we know that the contribution of $r$ to $|C_1|+\frac{1}{2}|C_{1/2}|$ becomes $1$, while the contribution of $c_t$ becomes $\frac{1}{2}$. Thus after iteration $t$ the total contribution of these two nodes to the sum is exactly $\frac{3}{2}$. However since $c_t,r \in O'$ we know that also $c_t,r \in O$ which means that $y_{c_t}^{c_t} + y_r^r \geq \frac{3}{4} + \frac{3}{4} = \frac{3}{2}$. Thus removing $c_t$ and $r$ from $O'$ simultaneously reduces the value of $\sum_{c \in O'} y_c^c$ by at least $\frac{3}{2}$. This implies that the entire value of $|C_1| + \frac{1}{2}|C_{1/2}| + \sum_{c \in O'} y_c^c$ does not increase and thus stays below $k$, which in turn proves that the claim is also fulfilled at the beginning of the $(t+1)$-th iteration of the loop.
\end{proof}

\begin{corollary}
    Algorithm \ref{Alg_h_integral} terminates in polynomial time.
\end{corollary}

To bound the sum of the cost of replacing all nodes in $O \cap C_{1/2}$ by their respective successors, we first need to bound the sum of the potentials of all nodes in $O$:

\begin{lemma}
\label{lem:potential_bound}
It holds that
\begin{equation*}
    \sum_{c \in O} \Phi_c \leq 13.5k \cdot \cost(\tilde{x}) + \frac{1}{8} \cost(y).
\end{equation*}
\end{lemma}

\begin{proof}
\begin{align*}
    \sum_{c \in O} \Phi_c &\stackrel{\phantom{\text{Lem. }\ref{lem:decrease}}}{=} \sum_{c \in O} \left(3 M_c \cdot r_c + \frac{1}{16k} \cost(y)\right)\\
    &\stackrel{\phantom{\text{Lem. }\ref{lem:decrease}}}{\leq} 3 \left( \sum_{c \in C} M_c \cdot r_c \right) + \frac{2k}{16k}\cost(y)\\
    &\stackrel{\phantom{\text{Lem. }\ref{lem:decrease}}}{=} 3 \left(\sum_{c \in C} \sum_{v \in V} y_v^c r_c \right) + \frac{1}{8} \cost(y)\\
    &\stackrel{\text{Lem. }\ref{lem:shift_by_radius}}{\leq} 13.5k \cdot \cost(\tilde{x}) + \frac{1}{8} \cost(y)
\end{align*}
\end{proof}

Next we will bound for any $c \in O$ the minimum replacement cost by any other center $\tilde{c} \in C \setminus \{c\}$, i.e. $\min_{\tilde{c} \in C \setminus \{c\}} R(c,\tilde{c})$, by $\frac{\Phi_c}{a_c}$. To do this, we need to identify assignments of $c$ to other centers of a total amount of $a_c$ such that their entire (fractional) length is at most $3 r_c$. Intuitively these assignments result from the fact that if for a node $c' \in V$ the initial assignment $\tilde{x}_c^{c'}$ of $c$ to $c'$ is larger than the value $y_c^{c',c}$ of opening shifted from $c'$ to $c$ at the end of Algorithm \ref{Alg_half} then an opening of value at least $\tilde{x}_c^{c'} - y_c^{c',c}$  has been shifted from $c'$ to other centers. The length of the corresponding assignments of $c$ to these centers can then be bounded by the previous corollary. In the following lemma these assignments are denoted by $f_c^{\tilde{c}}$ for any $\tilde{c} \in C \setminus \{c\}$.

\begin{lemma}
\label{lem:ugly}
For any $c \in O$ there exist values $(f_c^{\tilde{c}})_{\tilde{c} \in C \setminus \{c\}}$ such that:
\begin{enumerate}
\item $\forall \tilde{c} \in C \setminus \{c\}: f_c^{\tilde{c}} \leq y_c^{\tilde{c}}$,
\item $\sum\limits_{\tilde{c} \in C \setminus \{c\}} f_c^{\tilde{c}} = a_c$,
\item $\sum\limits_{\tilde{c} \in C \setminus \{c\}} f_c^{\tilde{c}} \cdot d(c,\tilde{c}) \leq 3 r_c$.
\end{enumerate}
\end{lemma}

\begin{proof}

One might observe that at the beginning of Algorithm \ref{Alg_half} the total assignment $\sum_{c' \in V} x_c^{c'}$ of $c$ by $x$ was equal to $m_c = \sum_{c' \in V} \tilde{x}_c^{c'}$ and that it gets decreased to $0$ during its execution. Additionally for any $c' \in V$ it holds that for any shift from $c'$ to $c$ the decrease of $x_c^{c'}$ during this shift is upper bounded by the increase of $y_c^{c',c}$.
To see this, consider a shift from $c'$ to $c$ over $v$, where potentially $v=c$.
By observation \ref{obs:value_of_x} $x_v^{c'} = \Delta^{\tilde{x}} (D_{c'},c,c')$. Thus $x_v^{c'}$ tells us how much $\sep^{\tilde{x}}(D_{c'},c')$ increases when we add $v$ to $D_{c'}$. Since at the same time the value $\sep^{\tilde{x}}(D_{c'} \cup\{c\},c')$ does not decrease when we add $v$ to $D_{c'}$ and $x_c^{c'}$
is always equal to $\Delta^{\tilde{x}}(D_{c'},c,c')$ we can conclude that $x_c^{c'}$ decreases by at most $x_v^{c'}$ during this shift. Since at the same time the algorithm increases $y_c^{c',c}$ by $x_v^{c'}$ during this we get that the decrease of $x_c^{c'}$ is upper bounded by this increase.

The idea is now basically to set for any $\tilde{c} \in C \setminus \{c\}$ the value of $f_c^{\tilde{c}}$ to the amount by which the assignments of $c$ get decreased by shifts to $\tilde{c}$ and bound the respective distances by Corollary \ref{cor:center_shift}. Additionally since the total decrease can be larger than $a_c$ the resulting values need to be scaled down accordingly.

Let for all $\tilde{c} \in C $ the value $\delta_c^{c',\tilde{c}}$ be defined as the amount by which $x_c^{c'}$ is decreased by all shifts from $c'$ to $\tilde{c}$ during the algorithm. Then by the previous argument $\delta_c^{c',c}$ is upper bounded by $\tilde{x}_c^{c'}$ minus the final value of $ y_c^{c',c}$ while simultaneously $\sum_{\tilde{c} \in C} \delta_c^{c',\tilde{c}} = \tilde{x}_c^{c'}$. Let $\tilde{c}$ be any other center of $C \setminus \{c\}$ such that $\delta_c^{c',\tilde{c}} > 0$. It is easy to observe that $u_c^{c',\tilde{c}} \geq \delta_{c}^{c', \tilde{c}}$ since $u_c^{c',\tilde{c}} = x_c^{c'}$ before the first shift from $c'$ to $\tilde{c}$. Furthermore if a shift to $\tilde{c}$ reduces the value of $x_c^{c'}$ by some value it increases $y_{\tilde{c}}^{c',\tilde{c}}$ by at least this value  by a similar argument as above. Thus after the last shift from $c'$ to $\tilde{c}$ also $y_{\tilde{c}}^{c',\tilde{c}} \geq \delta_{c}^{c',\tilde{c}}$ and thus $y_c^{c',\tilde{c}} = \min(u_c^{c',\tilde{c}} ,y_{\tilde{c}}^{c',\tilde{c}}) \geq \delta_{c}^{c', \tilde{c}}$. For any $\tilde{c} \in C$ let $\delta_c^{\tilde{c}} = \sum_{c' \in V}\delta_{c}^{c', \tilde{c}}$ and if $\tilde{c} \neq c$ set $f_c^{\tilde{c}} = \frac{a_c}{m_c - \delta_c^c} \delta_{c}^{\tilde{c}}$. Then 1. is fulfilled as $y_c^{\tilde{c}} \geq \delta_c^{\tilde{c}}$ and $a_c = 1 - y_c^c\leq m_c -y_c^c \leq m_c - \delta_c^c$.


Additionally also 2. is satisfied:

\begin{align*}
    \sum_{\tilde{c} \in C \setminus \{c\}} f_c^{\tilde{c}} &= \frac{a_c}{m_c - \delta_c^c}\sum_{\tilde{c} \in C \setminus \{c\}} \delta_{c}^{\tilde{c}}\\
    &= \frac{a_c}{m_c - \delta_c^c} (m_c - \delta_c^c)\\
    &= a_c
\end{align*}

Using Corollary \ref{cor:center_shift} we obtain that $d(c,\tilde{c}) \leq 2 d(c,c') + 4 r_c$ for any $\tilde{c} \in C, c' \in V$ with $ \delta_c^{c',\tilde{c}} >0$. Thus:
\begin{align*}
    \sum_{\tilde{c} \in C} f_c^{\tilde{c}} d(c,\tilde{c}) &\stackrel{\phantom{\text{Cor. }\ref{cor:center_shift}}}{=} \frac{a_c}{m_c - \delta_c^c} \sum_{c' \in V} \sum_{\tilde{c} \in C \setminus \{c\}} \delta_c^{c',\tilde{c}} d(c,\tilde{c})\\
    &\stackrel{\text{Cor. }\ref{cor:center_shift}}{\leq} \frac{a_c}{m_c - \delta_c^c} \sum_{c' \in V} \left( \sum_{\tilde{c} \in C \setminus \{c\}} \delta_c^{c',\tilde{c}}\right) \left(2 d(c,c') + 4 r_c\right)\\
    &\stackrel{\phantom{\text{Cor. }\ref{cor:center_shift}}}{=} \frac{a_c}{m_c - \delta_c^c} \sum_{c' \in V} \left(\tilde{x}_{c}^{c'} - \delta_c^{c',c}\right) \left(2 d(c,c') + 4 r_c\right)\\
    &\stackrel{\phantom{\text{Cor. }\ref{cor:center_shift}}}{\leq} \frac{1- y_c^c}{m_c - \delta_c^c} \left(\sum_{c' \in V} \tilde{x}_c^{c'} 2 d(c,c')\right)  + \frac{a_c}{m_c - \delta_c^c} (m_c - \delta_c^c) 4 r_c\\
    &\stackrel{\phantom{\text{Cor. }\ref{cor:center_shift}}}{\leq}  \frac{1- y_c^c + \delta_c^c}{m_c} \left(\sum_{c' \in V} \tilde{x}_c^{c'}  2 d(c,c')\right) + a_c 4 r_c\\
    &\stackrel{\phantom{\text{Cor. }\ref{cor:center_shift}}}{\leq} 2 \frac{1}{m_c} \left(\sum_{c' \in V} \tilde{x}_c^{c'}  d(c,c')\right) + a_c 4 r_c\\
    &\stackrel{\phantom{\text{Cor. }\ref{cor:center_shift}}}{\leq} 3 r_c,
\end{align*}
where in the second last inequality, we used $y_c^c \geq \delta_c^c$ and in last inequality we used the definition of $r_c$ and that $a_c \leq \frac{1}{4}$ since $c \in O$.
\end{proof}

\begin{lemma}\label{lem:replacement-potential-bound}
    For any $c \in O$ it holds that
    \begin{equation*}
        \min_{\tilde{c} \in C  \setminus \{c\}} R(c,\tilde{c}) \leq \frac{\Phi_c}{a_c}.
    \end{equation*}
\end{lemma}

\begin{proof}
    For any $c \in O$ let the values of $f_c^{\tilde{c}}$ for any $\tilde{c} \in C \setminus \{c\}$ be defined as in Lemma \ref{lem:ugly}. Let $A \subseteq C \setminus \{c\}$ be the set of centers $\tilde{c}$ such that $f_c^{\tilde{c}} > 0$. Then we can bound $R(c, \tilde{c})$ for any $\tilde{c} \in A$ as follows:
    \begin{align*}
        R(c,\tilde{c}) &\leq \sum_{v \in V} y_v^c \cdot (d(v,\tilde{c}) - d(v,c)) + \sum_{v \in V} y_v^{\tilde{c}} \cdot d(v, \tilde{c}) \frac{1}{y_c^{\tilde{c}}} \frac{1}{16k}\\
        &\leq M_c \cdot d(c,\tilde{c}) + \sum_{v \in V} y_v^{\tilde{c}} \cdot d(v, \tilde{c})\frac{1}{f_c^{\tilde{c}}} \frac{1}{16k}
    \end{align*}
    Let now $c^*$ be the center in $C \setminus \{c\}$ minimizing $R(c,c^*)$. It holds that:
    \begin{align*}
        R(c,c^*) &\stackrel{\phantom{\text{Lem. }\ref{lem:decrease}}}{\leq} \min_{\tilde{c} \in A} M_c \cdot d(c,\tilde{c}) + \sum_{v \in V} y_v^{\tilde{c}} \cdot d(v, \tilde{c}) \frac{1}{f_c^{\tilde{c}}} \frac{1}{16k}\\
        &\stackrel{\text{Lem. }\ref{lem:ugly}}{\leq} \frac{1}{a_c} \sum_{\tilde{c} \in A} f_c^{\tilde{c}} \left( M_c \cdot d(c,\tilde{c}) + \sum_{v \in V} y_v^{\tilde{c}} \cdot d(v, \tilde{c}) \frac{1}{f_c^{\tilde{c}}} \frac{1}{16k} \right)\\
        &\stackrel{\phantom{\text{Lem. }\ref{lem:decrease}}}{=} \frac{1}{a_c} \left( M_c \sum_{\tilde{c} \in A} f_c^{\tilde{c}} d(c,\tilde{c})  + \frac{1}{16k} \sum_{\tilde{c} \in A} y_v^{\tilde{c}} d(v, \tilde{c})\right)\\
        &\stackrel{\text{Lem. }\ref{lem:ugly}}{\leq} \frac{1}{a_c} \left(  M_c \cdot 3 r_c + \frac{1}{16k} \cost(y)\right)\\
        &\stackrel{\phantom{\text{Lem. }\ref{lem:decrease}}}{=} \frac{\Phi_c}{a_c}
    \end{align*}
    which proves the lemma.
\end{proof}

Using this we can bound the replacement cost of the centers in $O \cap C_{1/2}$ by twice the amount of the total potential:

\begin{lemma}
At the end of Algorithm \ref{Alg_h_integral} it holds that:
    \begin{equation*}
        \sum_{c \in \left( C_{1/2} \cap O \right)} R(c,s(c)) \leq 2 \sum_{c \in O} \Phi_c
    \end{equation*}
\end{lemma}

\begin{proof}
    Let $O_{1/2} = O \cap C_{1/2}$ and $O_1 = O \cap C_1$. We will define a function $g: O_{1/2} \times O_1\rightarrow \mathbb{R}_{\geq 0}$ where we require for any $(c,\tilde{c}) \in O_{1/2} \times O_1$ that if $g(c,\tilde{c}) > 0$ that then also $\frac{\Phi_c}{a_c} \leq \frac{\Phi_{\tilde{c}}}{a_{\tilde{c}}}$. Additionally we require for any $c \in O_{1/2}$ that $\sum_{\tilde{c} \in O_1} g(c,\tilde{c}) = \frac{1}{2} - a_c$ and for any $\tilde{c} \in O_1$ that $\sum_{c \in O_{1/2}} g(c,\tilde{c}) \leq a_{\tilde{c}}$. Intuitively speaking if we interpret $C_1$ as the centers whose opening has been increased to $1$ and $C_{1/2}$ as those whose opening has been decreased to $\frac{1}{2}$ then $g$ is basically telling us how the openings have been shifted from the nodes in $O_{1/2}$ to the nodes in $O_1$.

    To define $g$ consider for any $c \in O_{1/2}$ the iteration of the algorithm it gets added to $C_{1/2}$. If we also added a node $r$ to $C_1$ we set $g(c,r) = a_r$. This does not violate the desired properties of $g$ because $c$ minimizes $\frac{\Phi_c}{a_c}$ among all centers in $O'$ and $r$ was still an element of $O'$ at the beginning of this iteration. Additionally since $c,r \in O$ it holds that $\frac{1}{2} - a_c \geq \frac{1}{2} - \frac{1}{4} = \frac{1}{4} \geq a_r$. Besides this we will define for the centers which are in $O'$ after the end of the while loop (which will be added to $O_1$) an arbitrary mapping $g(c,\tilde{c})$ from the nodes in $O_{1/2}$ to them such that in total every $c \in C_{1/2}$ is mapped with a value of $\frac{1}{2} -a_c$ to $O_1$ and the total mapping to any $\tilde{c}\in O'$ is at most $a_{\tilde{c}}$. To see that such a mapping exists, first verify that $|O'| + |C_1| + \frac{1}{2} |C_{1/2}| = k$ after the last iteration of the while loop since the value of the sum only decreases by $\frac{1}{2}$ in every iteration. If the term was smaller then $k$ at the beginning of the while loop, then we have $O_{1/2} = \emptyset$ and the lemma is trivially fulfilled. Secondly combined with the fact that the total opening is bounded by $k$ implies:
    \begin{align*}
        \sum_{c \in O_{1/2}} \left(\frac{1}{2} - a_c\right)  - \sum_{\tilde{c} \in O_1 \setminus O'} g(c,\tilde{c}) &= \frac{1}{2} |O_{1/2}| - \sum_{c \in O \setminus O'} a_c\\
        &= \frac{1}{2} |O_{1/2}| - \sum_{c \in O \setminus O'} 1 - y_c^c\\
        &= \sum_{c \in O \setminus O'} y_c^{c} + \frac{1}{2} |O_{1/2}| - |O \backslash O'|\\
        &= \sum_{c \in O \setminus O'} y_c^{c} - \frac{1}{2} |O_{1/2}| - |O_1 \setminus O'|\\
        &\leq \sum_{c \in O \setminus O'} y_c^{c} + \sum_{c \in H} y_c^c - \frac{1}{2}|H| - \frac{1}{2} |O_{1/2}| - |O_1 \setminus O'|\\
        &= \sum_{c \in C \setminus O'} y_c^{c} -k + k - \frac{1}{2} |C_{1/2}| - |O_1 \setminus O'|\\
        &\leq - \sum_{\tilde{c} \in O'} y_{\tilde{c}}^{\tilde{c}} + |O'| \\
        &= \sum_{\tilde{c} \in O'} a_{\tilde{c}}
    \end{align*}

    Using this mapping we can basically redistribute the potential of the centers added to $O_1$ to the centers in $O_{1/2}$ to bound for any $c \in O_{1/2}$ the value of $R(c,s(c))$ as follows:
    \begin{alignat*}{3}
        R(c,s(c)) &\stackrel{\text{Lem. }\ref{lem:replacement-potential-bound}}{\leq}&&\frac{\Phi_c}{a_c}\\
        &\stackrel{\phantom{\text{Lem. }\ref{lem:replacement-potential-bound}}}{=}&& \Phi_c \left( 2 + 2\frac{ 0.5 - a_c}{a_c}\right)\\
        &\stackrel{\phantom{\text{Lem. }\ref{lem:replacement-potential-bound}}}{=}&& \Phi_c + 2 \sum\limits_{\tilde{c} \in O_1} g(c,\tilde{c}) \frac{\Phi_c}{a_c}\\
        &\stackrel{\phantom{\text{Lem. }\ref{lem:replacement-potential-bound}}}{\leq}&&  \Phi_c + 2 \sum\limits_{\tilde{c} \in O_1} g(c,\tilde{c}) \frac{\Phi_{\tilde{c}}}{a_{\tilde{c}}}
    \end{alignat*}

    And using that for any $\tilde{c} \in O_1$ the amount by which the other centers are mapped to it is upper bounded by $a_{\tilde{c}}$ we can thus bound the reassignment cost of all nodes in $O_{1/2}$:
    \begin{align*}
        \sum_{c \in O_{1/2}} R(c,\tilde{c}) &\leq \sum_{c \in O_{1/2}} \left(2 \Phi_c + 2 \sum_{\tilde{c} \in O_1} g(c,\tilde{c}) \frac{\Phi_{\tilde{c}}}{a_{\tilde{c}}}\right)\\
        &= 2 \sum_{c \in O_{1/2}} \Phi_c + 2   \sum_{\tilde{c} \in O_1} \sum_{c \in O_{1/2}} g(c,\tilde{c}) \frac{\Phi_{\tilde{c}}}{a_{\tilde{c}}}\\
        &\leq 2 \sum_{c \in O_{1/2}} \Phi_c + \sum_{\tilde{c} \in O_1} a_{\tilde{c}} \frac{\Phi_{\tilde{c}}}{a_{\tilde{c}}}\\
        &= 2 \sum_{c \in O} \Phi_c
    \end{align*}
\end{proof}

Combining this with Lemma \ref{lem:potential_bound} directly gives us an $O(k)$ bound for the total cost necessary if we replaced all nodes in $O \cap C_{1/2}$ by their successors:

\begin{corollary}
\label{cor:replace_o_bound}
    \begin{equation*}
        \sum_{c \in \left( C_{1/2} \cap O \right)} R(c,s(c)) \leq 27k \cdot \cost(\tilde{x}) + \frac{1}{4} \cost(y)
    \end{equation*}
\end{corollary}

\subsection{Finding the integral centers}

After the previous step we have subdivided our centers into $C_{1/2}$ and $C_1$ with $|C_1| + \frac{1}{2} |C_{1/2}| \leq k$ and for each $c \in C_{1/2}$ we have a successor $s(c)$ such that replacing the assignments to $c$ by assignments to $s(c)$ costs not too much. Now we would like to find a subset $F$ of $C_{1/2}$ of size at most $k - |C_1|$ such that for every center $c \in C_{1/2}$ either $c$ or $s(c)$ is contained in $F \cup C_1$, which will be our final center set.

To do this we consider the directed graph $G_{1/2} = (C_{1/2},S)$ with $S = \big\{ (c,s(c)) \mid c,s(c) \in C_{1/2}\big\}$. In the original LP rounding algorithm for $k$-median the authors could use that this graph is bipartite since there exist no directed cycles of uneven length however in our case this is not true. But since there is only one outgoing edge for each center we know that each component of the graph contains at most one cycle. Additionally because of the if-statement in line \ref{line:if} of Algorithm \ref{Alg_h_integral} we know that any cycle in $G_{1/2}$ contains at most one node from $O$. For any cycle of uneven length let $c$ be the singular node in it contained in $O$ if such a node exists and an arbitrary node otherwise. Let $p$ be the predecessor of $c$ on the cycle and $g$ the predecessor of $p$. Both $p$ and $g$ are contained in $H$ because the shortest uneven cycle length is $3$.

We will now show that either it is possible to set $s(p)=g$ to eliminated the cycle, or to set $s(g)=c$ to shorten the cycle by one.
We distinguish two cases: If $d(g,p) \leq d(p,c)$ we turn the cycle into a tree by setting $s(p) = g$ instead of $c$. Since $g \in H$ we know by construction of Algorithm \ref{Alg_h_integral} that $y_g^p \geq \frac{1}{16k}$. Thus the graph would stay $\frac{1}{16k}$-connected if all assignments to $p$ will be reassigned to $g$ instead. Furthermore since $g$ is closer to $p$ than $c$ it still holds that $R(p,g) \leq M_p \cdot 16 r_p$. Otherwise we set $s(g) = c$. This way the cycle length reduces by 1 and the component gets bipartite. As in the previous case it holds that $y_g^p \geq \frac{1}{16k}$ and $y_p^c \geq \frac{1}{16k}$. 
Note, that this does not necessarily guarantees us, that there is a flow from $g$ to $c$ of value $\frac{1}{16k}$ in $(y_v^c)_{v \in V}$. This only implies that there is such a flow from $g$ to $p$ in $(y_v^p)_{v \in V}$ and a flow from $p$ to $c$ in $(y_v^c)_{v \in V}$ with this values. We observe, that if we replace $p$ by $c$, then we get our desired flow from $g$ to $c$ with value $\frac{1}{16k}$ and our connectivity is ensured. As a consequence, we can only reassign $g$ to $c$, if we also reassign $p$ to $c$. However since both $p$ and $g$ are adjacent to $c$ after this step, they will be part of the same half of the bipartition which directly guarantees us that we either replace both $p$ and $g$ or add both of them to the center set. Thus, the graph will stay $\frac{1}{16k}$-connected after the reassignments.
To bound the cost, by the triangle inequality we have that $d(g,c) \leq d(g,p) + d(p,c) \leq 32 r_g$. Thus $R(g,c) \leq M_g \cdot 32 r_g$.

Once we have made the graph bipartite this way we will consider a bipartition $H_1, H_2$ for every connected component. Let without loss of generality $|H_1| \leq |H_2|$. Then $H_1$ contains at most half as many centers as the entire component $H_1 \cup H_2$. Additionally for any $c \in H_2$ it holds that $s(c) \in H_1$ or $s(c) \in C_1$. Thus by defining $F$ to be the union of all the smaller halves of the bipartition of the connected components ensures that for all $c \in C_{1/2} \setminus F$  that $s(c) \in C_1 \cap F$ while $|F| \leq \frac{|C_{1/2}|}{2} \leq k - |C_1|$. Then we set $C_1 = C_1 \cup F$ and use it as our final center set. Let $C_0 = C\setminus C_1$. Then we can define a $\frac{1}{16k}$-connected solution $\left(z_v^c\right)_{v \in V,c \in C_1}$ using the following two steps: 
\mycomment{
{\color{red}
\begin{itemize}
\item For all $c \in C_1$:
\begin{itemize}
\item $z_c^c = 1$
\item $\alpha_c = \max\left(1,\max_{\tilde{c} \in C_0 \cap O: s(\tilde{c}) = c} \frac{1}{16k}\frac{1}{y_{\tilde{c}}^{c}}\right)$
\item $\forall v \in V \setminus \{c\}$ $z_v^c = \alpha_c \cdot y_v^c$
\end{itemize}
\item For all $c \in C_0$:
\begin{itemize}
    \item Let $\tilde{c} = s(c)$
    \item $\forall v \in V$: $z_v^{\tilde{c}} = z_v^{\tilde{c}} + y_v^c$
\end{itemize}
\end{itemize}
}
}
\begin{itemize}
\item For all $c \in C_1$ we set $z_c^c = 1$ and then define a value $\alpha_c = \max_{\tilde{c} \in C_0 \cap O: s(\tilde{c}) = c} \frac{1}{16k}\frac{1}{y_{\tilde{c}}^{c}}$. If this value is greater $1$ we set for all $v \in V \setminus \{c\}$ the assignment $z_v^c = \alpha_c \cdot y_v^c$ and otherwise $z_v^c = y_v^c$. This ensures for any $c \in O \cap C_0$ that $c$ is $\frac{1}{16k}$-connected to $s(c)$. As argued above for the centers in $H \cup C_0$ this is already fulfilled by the solution $y$.
\item Afterwards we  go over all $c \in C_0$ and increase for all nodes $v \in V$ the value $z_v^{s(c)}$ by $y_v^c$, i.e. we set $z_v^{s(c)} = z_v^{s(c)} + y_v^c$. This way we redirect all assignments to $c$ to its successor instead. Since we know that these successors are all contained in $C_1$ the resulting solution $z$ is indeed only using $|C_1|$ as a center set.
\end{itemize}

We can bound the cost of the solution $z$ as follows:

\begin{lemma}
\label{lem:bound_int_centers_z}
Let $z$ be the final solution using only integral centers. Then
\begin{equation*}
    \cost(z) \in O( k \cdot cost(\tilde{x})).
\end{equation*}
\end{lemma}

\begin{proof}
We may first observe that $\cost(z) \leq \cost(y) + \sum\limits_{c \in C_0} R(c,s(c))$. This is due to the fact that for every $c \in C_0$ the reassignments in the second step would also have been executed if we would have only replaced $c$ by $s(c)$. Additionally for any $c \in C_1$ with $\alpha_c > 1$ let $p_c = \argmax\left(1,\max_{\tilde{c} \in C_0 \cap O: s(\tilde{c}) = c} \frac{1}{16k}\frac{1}{y_{\tilde{c}}^{c}}\right)$. Then also during the replacement of $p_c$ by $c$ the assignments to $c$ would have been increased by $\frac{1}{\alpha_c}$. Thus any modification performed during the calculation of $z$ is also part of at least one replacement procedure which gives us the described bound. As a result:
\begin{alignat*}{2}
    \cost(z) &\mathmakebox[1.5cm][c]{ \stackrel{}{\leq} }&& \cost(y) + \sum\limits_{c \in C_0} R(c,s(c))\\
    &\mathmakebox[1.5cm][c]{\stackrel{}{\leq}}&& \cost(y) + \sum\limits_{c \in C_{1/2}} R(c,s(c))\\
    &\mathmakebox[1.5cm][c]{\stackrel{}{=}}&& \cost(y) + \sum\limits_{c \in H} R(c,s(c)) + \sum\limits_{c \in O \cap C_{1/2}} R(c,s(c))\\
    &\mathmakebox[1.5cm][c]{\stackrel{\text{Cor. }\ref{cor:replace_o_bound}}{\leq}}&&  \cost(y) + 27k \cdot \cost(\tilde{x}) + \frac{1}{4} \cost(y) + \sum\limits_{c \in H} M_c d(c,s(c))\\
    &\mathmakebox[1.5cm][c]{\stackrel{\text{Thm. }\ref{thm:bound_y}}{\leq}}&& 52k \cdot \cost(\tilde{x}) + \sum\limits_{c \in H}M_c 32 r_c\\
    &\mathmakebox[1.5cm][c]{\stackrel{}{\leq}}&& 52k \cdot\cost(\tilde{x}) + 32 \sum\limits_{c \in C}M_c r_c\\
    &\mathmakebox[1.5cm][c]{\stackrel{\text{Lem. }\ref{lem:shift_by_radius}}{\leq}}&& 196k \cdot \cost(\tilde{x})
\end{alignat*}
\end{proof}

\subsection{Getting an entirely integral solution}

If we obtained $z$ we can multiply its values by $16k$ (capped at $1$) which results in a 1-connected solution where every node is assigned to at least one of the centers with value $1$. Thus as in the assignment version we can use the primal-dual algorithm to obtain integer assignment values while increasing the cost only by $O(\log n)$. By combining this with Lemma \ref{lem:bound_int_centers_z} we obtain the following theorem:

\ClusteringNonDisjoint*


\subsection{Technical Lemmas}\label{sec:technical}


\begin{lemma}
\label{lem:cup_cap}
    For any $A,B\subseteq V$ and any $t \in V$, let:
    \begin{itemize}
    \item $N_{A,B}^\cap = (A \cap I_t(B)) \cup (B \cap I_t(A)) \cup (A\cap B)$
    \item $N^\cup_{A,B} = (A \setminus I_t(B)) \cup (B \setminus I_t(A))$
    \end{itemize}
    Then:
    \begin{enumerate}
    \item $I_t(N_{A,B}^\cap) = I_t(A) \cap I_t(B)$
    \item $H_t(N_{A,B}^\cap) = H_t(A) \cap H_t(B)$
    \item $I_t(N_{A,B}^\cup) \supseteq I_t(A) \cup I_t(B)$
    \item $H_t(N_{A,B}^\cup) \supseteq H_t(A) \cup H_t(B)$
    \end{enumerate}
\end{lemma}

\begin{proof}

To show 1. let $v \not\in I_t(A)$. Then there exists a path from $V$ to $t$ not passing through $H_t(A)$. Since it is easy to verify that $N_{A,B}^\cap \subseteq H_t(A)$ the path is also not passing through $N_{A,B}^\cap$ and $v\not \in I_t(N_{A,B}^\cap)$. Since the same argument also holds for $B$ we get that $I_t(N_{A,B}^\cap) \subseteq I_t(A) \cap I_t(B)$. Let now $v \in I_t(A) \cap I_t(B)$ and suppose that there existed a path from $v$ to $t$ not containing any nodes in $N_{A,b}^\cap$. Let $p$ be the first node on the path not in $I_t(A) \cap I_t(B)$. Then if $p \in I_t(A) \setminus I_t(B)$ we may observe that $p \in I_t(A) \cap B \subseteq N_{A,B}^\cap$. The case $p \in I_t(B) \setminus I_t(A)$ can be handled analogously. Thus the only remaining case is that $p\not\in I_t(A)$ and $p \not\in I_t(B)$ but then $p \in A \cap B \subseteq N_{A,b}^\cap$. Thus such a path cannot exist and $v \in I_t(N_{A,B}^\cap)$.

Additionally this implies 2., since:
\begin{align*}
H_t(A) \cap H_t(B) &= (A \cap B) \cup (A \cap I_t(B)) \cup (I_t(A) \cap B) \cup (I_t(A) \cup I_t(B))\\
&= N_{A,B}^\cap \cup (I_t(A) \cap I_t(B))\\
&= N_{A,B}^\cap \cup I_t(N_{A,B}^\cap) = H_t(N_{A,B}^\cap)
\end{align*}

To show 3. let $v$ be an arbitrary element of $I_t(A) \cup I_t(B)$. Suppose there exists a path from $v$ to $t$ not containing any elements from $N_{A,B}^\cup$. Let $p_1,p_2$ be successors on that path such that $p_1 \in I_t(A) \cup I_t(B)$ and $p_2 \not\in I_t(A) \cup I_t(B)$. W.l.o.g. $p_1 \in I_t(A)$ and thus $p_2 \in A$ which together with $p_2 \not \in I_t(B)$ implies that $p \in A \setminus I_t(B) \subseteq N_{A,B}^\cup$. Thus such a path cannot exist and $v \in I_t(N_{A,B}^\cup)$.

Again this directly implies 4.:
\begin{align*}
H_t(A) \cup H_t(B) &= I_t(A) \cup I_t(B) \cup A \cup B\\
&= (I_t(A) \cup I_t(B)) \cup (A \setminus I_t(B)) \cup (B \setminus I_t(A))\\
&\subseteq I_t(N_{A,B}^\cup) \cup N_{A,B}^\cup = H_t(N_{A,B}^\cup)
\end{align*}

\end{proof}

\mycomment{
\begin{lemma}
For any $s,t \in V$ and weight function $w$ it holds that if $A$ and $B$ are minimum weight cuts between $s$ and $t$ the same also holds for $A \cup B$.
\end{lemma}

\begin{proof}
For any $S, T \subseteq V$ let $\Gamma(S,T) = \{\{v,w\}\in E \mid v \in S, w \in T \setminus S$ and $w(S,T) = \sum_{e \in \Gamma(s,t)} w_e$. Then:

\begin{equation*}
    w(\Gamma(B)) = w(B,A) + w(B, V \setminus A)
\end{equation*}

\begin{equation*}
    w(\Gamma(A\cup B)) = w(A, V \setminus (B \cap \Gamma(B))) + w(B, V \setminus A)
\end{equation*}

If we assume that $A \cup B$ is not a minimum cut  while $A$ and $B$ are then $ w(\Gamma(A\cup B))> w(\Gamma(A)$ which implies:
\begin{equation*}
     w(A, V \setminus (B \cap \Gamma(B))) > w(B,A)
\end{equation*}

Using this we can bound $w(\Gamma(A))$ as follows:
\begin{align*}
    w(\Gamma(A)) &= w(A,B) + w(A, V \setminus B)\\
    &= w(A,B) + w(A, V \setminus (B\cup\Gamma(B)) + w(A, \Gamma(B)\\
    &> w(A,B) + w(B,A) + w(A, \Gamma(B))\\
    &\geq w(A \cap B, A) + w(A\cap B, B) + w(A\cap B, V \setminus (A \cup B))\\
    &= w(\Gamma(A \cap B))
\end{align*}

However this is a direct contradiction to $A$ being a minimum weight cut. Thus $A \cup B$ also needs to be a minimum cut.
\end{proof}

}

\begin{lemma}
\label{lem:sum_cup_cap}
For any $A, B \subseteq V$ it holds that $w(A) + w(B) \geq w(N_{A,B}^\cap) + w(N_{A,B}^\cup)$.
\end{lemma}

\begin{proof}
It holds that:
\begin{align*}
w(N^\cup_{A,B}) &= w((A \setminus I_t(B)) \cup (B \setminus I_t(A)))\\
&= w(A \setminus I_t(B)) + w(B \setminus I_t(A)) - w((A \setminus I_t(B)) \cap (B \setminus I_t(A)))\\
&= w(A \setminus I_t(B)) + w(B \setminus I_t(A)) - w(A \cap B)
\end{align*}

while at the same time by the definition of $N_{A,B}^\cap$:

\begin{align*}
w(N^\cap_{A,B}) &\leq w(A \cap I_t(B))  + w(B \cap I_t(A)) + w(A\cap B)
\end{align*}

And by combining these two bounds we are directly able to prove the lemma:

\begin{align*}
    w(A) + w(B) & = w(A \setminus I_t(B)) + w(A \cap I_t(B)) + w(B \setminus I_t(A)) + w(B \cap I_t(A))\\
    &\geq w(N^\cap_{A,B}) + w(N_{A,B}^\cup)
\end{align*}
\end{proof}

\begin{lemma}
For all $S,S' \subseteq V$, $t,v \in V$ with $S \subseteq S'$ it holds that:
\begin{equation*}
 \Delta^w(S',v,t) \leq \Delta^w(S,v,t)
\end{equation*}
\end{lemma}

\begin{proof}
    Let $A$ be a minimum cut separating $cut_H(S \cup \{v\}, t)$ and let $B$ be a minimum cut separating $cut_H(S',t)$. Then it holds that:
    \begin{align*}
        \sep^w(S' \cup \{v\},t) &\stackrel{\text{Lem. }\ref{lem:cup_cap}}{\leq} w(N_{A,B}^\cup)\\
        &\stackrel{\text{Lem. }\ref{lem:sum_cup_cap}}{\leq} w(A) + w(B) - w(N_{A,B}^\cap)\\
        &\stackrel{\text{Lem. }\ref{lem:cup_cap}}{\leq} \sep^w(S \cup \{v\},t) + \sep^w(S',t) - \sep^w((S \cup\{v\})\cap S',t)\\
        &\stackrel{\phantom{\text{Lem. }\ref{lem:decrease}}}{\leq} \sep^w(S \cup \{v\},t) + \sep^w(S',t) - \sep^w(S,t)\\
        &\stackrel{\phantom{\text{Lem. }\ref{lem:decrease}}}{=} \sep^w(S',t) + \Delta^w(S,v,t)
    \end{align*}
    And thus obviously $ \Delta^w(S',v,t) \leq \Delta^w(S,v,t)$.
\end{proof}

\begin{lemma}
For any $S, S' \subseteq V$, $t \in V$ and any cut $N$ between $S$ and $t$ it holds that:
\begin{equation*}
\sep^w(N \cup S',t) \geq \sep^w(S \cup S',t)
\end{equation*}
\end{lemma}

\begin{proof}
    Let $N'$ be a cut between $N$ and $t$ and let $v$ be an arbitrary element in $ H_t(N)$. Then any path from $v$ to $t$ must contain at least one node $p \in N$. Since $p \in H_t(N')$ and $t \not\in I_t(N')$ we may conclude that the path must also contain an element in $N'$. And since this holds for any path from $v$ to $t$ we get that $v \in H_t(N')$. Thus any cut between $N$ and $t$ must also be a cut between $S$ and $t$. In particular this also holds for all cuts $N'$ with $S' \subseteq H_t(N')$ which means that all cuts between $S' \cup N$ and $t$ are also cuts between $S' \cup S$ and $t$ and thus $\sep^w(N \cup S',t) \geq \sep^w(S \cup S',t)$.
\end{proof}

%% file: inapproxRatio_Non_disjoint.tex
\section{Hardness Result}
\label{sec:appendix_hardness}

\begin{theorem}[Hardness of Non-Disjoint Connected \(k\)-Median] \label{thm:non_disjoint_connected_k_median} 
There is a constant $c > 0$ such that for all sufficiently large $n$, the non-disjoint connected $k$-median problem cannot be approximated within a factor of $c \cdot \log n$ unless P=NP, even when $k=2$. The same is true for the assignment variant.\end{theorem}


We establish this hardness via an approximation‐preserving reduction from the dominating set problem, which is known to be NP‐hard to approximate within a factor of $c\log n$ for some constant $c $ \cite{alon2006algorithmic}.  
The proof (and reduction) mainly follows the one presented by Gupta et al. \cite{gupta2011clustering}, but the relaxed ``non‐disjoint clusters'' constraints require more careful analysis. In particular, while preserving the overall construction, we remove certain edges to prevent the formation of advantageous non‐disjoint clusters that would otherwise arise from them. 

Given a graph $G = (V,E)$ with $V = (v_1,..,v_n)$, a set $S \subseteq  V$ is called a \emph{dominating set} if for every $i \in [n]$ the node $v_i$ itself or a neighbor $v_j$ of $v_i$ is contained in $S$. The aim of the dominating set problem is to find a dominating set of minimum size. For our reduction we construct a connectivity graph $G' = (V',E')$ as follows:
\begin{itemize}
    \item $V' =\{a,b,x_1,\ldots,x_n,y_1,\ldots,y_n\}$
    \item For all $i \in \{1,\ldots,n\}$ we add the edges $\{x_i,a\}$, $\{x_i,b\}$ and $\{x_i,y_i\}$ to $E'$
    \item for all edges $\{v_i,v_j\} \in E$ we additionally add the edges $\{x_i,y_j\}$ and $\{x_j,y_i\}$ to $E'$.
\end{itemize}
As for a distance metric $d$ we have two positions $0$ and $1$ in the one dimensional Euclidean space and place the nodes $a$, $x_1,\ldots,x_n$ at position $0$ and the nodes $b,y_1,\ldots,y_n$ at position $1$. The distance between two nodes is then the Euclidean distance between the respective positions, resulting in a distance of $0$ if two nodes are placed at the same position and a distance of $1$ otherwise. Figure~\ref{figure:hardness_nondisjoint} illustrates the resulting structure for a given example. 

\begin{figure}
\begin{subfigure}{0.3\textwidth}
\centering
    \includegraphics[width= 0.8\textwidth]{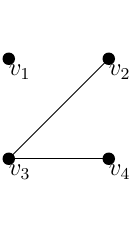}
\end{subfigure}
\hfill
\begin{subfigure}{0.6\textwidth}
\centering
    \includegraphics[width= 0.8\textwidth]{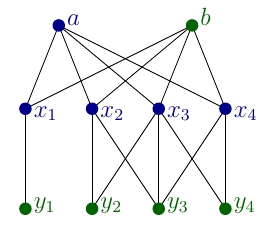}
\end{subfigure}
\caption{An example of the transformation of a graph $G$ acting as a dominating set instance (left) into a connected $2$-median instance (right) performed by our reduction. The blue nodes in the right graph are placed at position $0$ and the green ones at position $1$.}
\label{figure:hardness_nondisjoint}
\end{figure}


Following the same logic as in the reduction for the disjoint setting presented by Gupta et al., any (potentially non-disjoint) $2$-median solution on this instance with objective value $s$ directly corresponds to a dominating set of size $s$ and vice versa. In the following we present a formal proof of this since Gupta et al. omitted it in their work:

\begin{lemma}
   There exists a dominating set of size $s$ in $G$ if and only if there exists a non-disjoint connected $2$-median solution with objective value $s$ for the connectivity graph $G'$ and distance function $d$.
\end{lemma}
\begin{proof}

If we are given a dominating set $S$, we can construct a clustering solution with cost $|S|$ as follows: We choose $a$ and $b$ as centers;  for all $i \in [n]$ we assign $y_i$ to $b$; for all $i \in [n]$ we assign $x_i$ to $b$ if $v_i \in S$ and to $a$ otherwise. It is easy to verify that every node $y_i$ is indeed connected to $b$, as there exists a node $v_j \in S$ such that either $v_j = v_i$ or $v_i$ is neighbored to $v_j$ which means that in $G'$ the node $x_j$ (corresponding to $v_j$) is neighbored to $y_i$ and gets assigned to $b$. Additionally the cost of this solution is equal to the number of nodes in $\{x_1,\ldots,x_n\}$ assigned to $b$ which is exactly $|S|$.
 
Conversely: A small \(k\)-median cost implies a small dominating set. For any given k-median solution with cost $s<n$, we can make the following observations: 
\begin{itemize}
    \item One center must be placed at position $0$ while the other one needs to be placed at position $1$. We denote the respective centers as $c_0$ and $c_1$.
    \item Since $a$ is neighbored to any node placed at position $0$, it is in particular neighbored to $c_0$ and we can always assign it to $c_0$ without incurring any additional costs.
    \item Since there are $n+1$ nodes at position $1$ and all of them are only neighbored to nodes in $\{x_1,\ldots,x_n\}$, we know that there exists an $i \in [n]$ such that $x_i$ needs to be assigned to $c_1$. Since $b$ is neighbored to $x_i$, we can assign it to $c_1$ without incurring any additional costs.
\end{itemize} 
Using these observations we may assume w.l.o.g. that $c_0 = a$ and $c_1 = b$. Additionally we can ensure that for every $i \in [n]$ the node $y_i$ is assigned to $b$ without increasing the cost of the solution. Suppose that $y_i$ is assigned to $a$ then $y_i$ is contributing a value of $1$ to the cost of the solution. By reassigning it to $b$ its contribution decreases to $0$ however this may violate our connectivity constraint. To avoid this, we can also assign $x_i$ to $b$ which will then again contribute a value of $1$ to the cost, meaning that the entire cost does not increase. 
By applying this process to every $ y_i $ initially assigned to $ a $, we obtain a solution that assigns all nodes in  $\{y_i\}_{i \in [n]}$ to $b$ without increasing the cost. Let $B$ be the set of all nodes assigned to $b$ in this solution. Notice that for every node $v_i \in V$ it holds that $y_i \in B$ and thus $y_i$ needs to be connected to $b$. 

We consider the subset $S = \{v_i\mid i \in [n] \land x_i \in B\}$ of $V$. 
For every $i \in [n]$ we know that $y_i$ is neighbored to a node $x_j \in B$ which by the definition of $G'$ means that either $i = j$ or $v_i$ and $v_j$ are neighbored in $G$.
Thus either $v_i$ itself or a neighbor of $v_i$ is contained in $S$. Since this is true for any node in $V$, the set $S$ forms a dominating set of $G$ whose size is equal to the number of nodes in $\{x_i\}_{i \in [n]}$ contained in $B$ which is upper bounded by $s$.
\end{proof}

Since this reduction is approximation-preserving, the hardness of Dominating Set also carries over to the non-disjoint variant of connected $k$-median. Since the reduction also works if we specify that $a$ and $b$ are chosen as centers this is also true for the assignment version.

%% file: disjoint_general.tex

If we require the clusters to be disjoint, the connected $k$-median problem gets even harder to approximate. We will prove via a reduction from the 3-SAT problem that even for $k = 2$ we cannot hope for an approximation guarantee much better than $O(n)$, unless P = NP. The reduction resembles the hardness proof for the connected $k$-center problem in~\cite{Drexler2024}.

\hardnessdisjoint*

Consider an instance of the well-known NP-hard $3$-SAT problem with variables $x_1,\ldots,x_a$ and a set of clauses $c_1,\ldots,c_b$. Each clause consists of up to three literals $x_i$ or $\overline{x}_i$. The goal is to decide whether or not there is an assignment of the variables that satisfies all clauses.

We will now reduce the given instance of the $3$-SAT problem to an instance of the connected $k$-median problem. For ease of presentation, we allow different nodes in this instance to have a distance of $0$. The reduction is, however, still valid if all these zero distances are replaced by sufficiently small values greater than $0$. There will be three different groups of nodes, namely $L$, $M$, and $R$. All nodes that belong to the same group have a pairwise distance of~$0$, while the distance between any node in $L\cup R$ and any node in $M$ is~$1$. Finally, the pairwise distances between nodes in $L$ and $R$ is~$2$. Formally this metric (also with non-zero instead of zero distances) can be encoded as a tree metric.

\begin{figure}[ht]
    \centering
    \includegraphics[width=8cm]{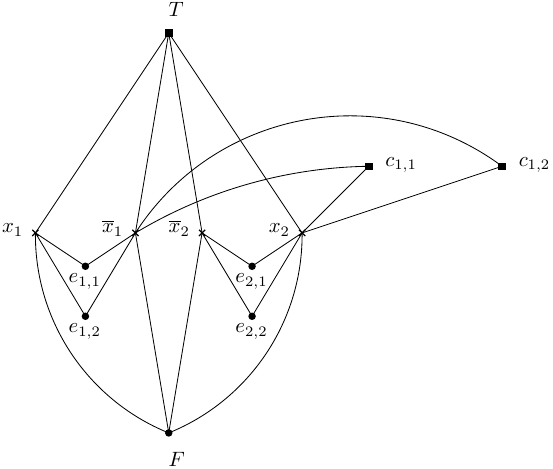}
    \caption{The resulting connectivity graph for the formula $(\overline{x}_1 \lor x_2)$ and $m = 2$. The position of the nodes is shown by their shape: a box represents  $L$, a cross represents $M$ and a disc represents $R$.}
    \label{fig:sketch_formula}
\end{figure}

Let now $m$ be a natural number whose value we will fix later. The set of nodes is constructed as follows:
\begin{itemize}
  \item At position $L$ we place a special node $T$ as well as $m$ nodes $c_{i,1},\ldots,c_{i,m}$ for every clause $c_i$. 
  \item At position $M$ we place two nodes $x_i$ and $\overline{x}_i$ for every variable $x_i$.
  \item At position $R$ we place another special node $F$ as well as $m$ nodes $e_{i,1},\ldots,e_{i,m}$ for every variable $x_i$.
\end{itemize}

Finally, we have to specify the connectivity graph. We add the following edges to this graph.
\begin{itemize}
 \item We connect $T$ and $F$ with $x_i$ and $\overline{x}_i$ for every $i \leq a$.
 \item We add edges between $x_i$ and $e_{i,j}$ as well as between
  $\overline{x}_i$ and $e_{i,j}$ for any $i \leq a$ and $ j \leq m$.
  \item For every literal $l$ and $i \leq b$, $j \leq m$ we add an edge between the node corresponding to $l$ and $c_{i,j}$ if $l$ is contained in the clause $c_i$.
\end{itemize}
 In Figure \ref{fig:sketch_formula} the resulting connectivity graph for the formula $(\overline{x}_1 \lor x_2)$ is depicted.

We will now prove the following two lemmas, showing that the optimal solution of the constructed instance of the connected $k$-median problem is significantly cheaper if the given 3-SAT instance is satisfiable.

\begin{lemma}\label{lem:dir1}
If the given 3-SAT instance is satisfiable, there exists a solution with cost $2a$ for the constructed instance of the connected $k$-median problem.
\end{lemma}

\begin{proof}
Let $\alpha_1,\ldots,\alpha_a$ be a satisfying assignment. We choose $T$ and $F$ as centers. If $\alpha_i = 1$, we assign $x_i$ to $T$ and $\overline{x}_i$ to $F$. If $\alpha_i = 0$, we assign $\overline{x}_i$ to $T$ and $x_i$ to $F$. Furthermore for any $i \leq a$ and $j \leq m$ we assign $e_{i,j}$ to $F$ and for any $i \leq b$ and $j \leq m$ we assign $c_{i,j}$ to $T$. 
In both cases, assigning $x_i$ and $\overline{x}_i$ each costs $1$ because $T$ and $F$ are in $L\cup R$ while $x_i$ and $\overline{x}_i$ are in $M$, and the distance between points in $L \cup R$ and $M$ is always $1$. Since there are $a$ variables, the cost of this assignment is $2a$.
It remains to show that the clusters are actually connected.

First let us consider the cluster of $F$. The only nodes assigned to $F$ that are not directly neighbored to it are nodes of the form $e_{i,j}$. Note that for every $i$ either $x_i$ or $\overline{x}_i$ is assigned to $F$. Since $e_{i,j}$ is neighbored to both of these nodes we may thus conclude that there exists a path from $e_{i,j}$ to $F$ within the cluster. Thus the entire cluster is connected.

Now we consider the cluster of $T$. Similarly as above only nodes of the form $c_{i,j}$ could cause problems regarding the connectivity. Let us consider the node $c_{i,j}$ for a fixed $i \leq b$, $j \leq m$. Since the clause $c_i$ is fulfilled by $\alpha$ there exist two cases:
\begin{itemize}
\item There exists an $i'$ such that $\alpha_{i'} = 1$ and $x_{i'} \in c_i$. Then $x_{i'}$ has been assigned to $T$ and is a neighbor of $c_{i,j}$ for any $j$. It follows that $c_{i,j}$ is connected to $T$.
\item There exists an $i'$ such that $\alpha_{i'} = 0$ and $\overline{x}_{i'} \in c_i$. Again, due to the fact that $\overline{x}_{i'}$ has been assigned to $T$ and is a neighbor of $c_{i,j}$ for any $j$ , we can conclude that $c_{i,j}$ is connected to $T$.
\end{itemize}
Thus we can conclude that both clusters are connected and our solution is feasible.
\end{proof}

\begin{lemma}\label{lem:dir2}
If there is no satisfying assignment of the 3-SAT formula, then there is no solution of the $k$-median instance with cost smaller than $2m$ (given that $a,b \geq 2$).
\end{lemma}

\begin{proof}
Let a solution with clustering cost smaller than $2m$ be given. Since $a,b \geq 2$, both in group $L$ and $R$ are at least $2m$ nodes. Thus the solution must have chosen a point at $L$ and a point at $R$ as a center. Since the neighborhood of any point at $L$ is a subset of the neighborhood of $T$ and the same holds for $R$ and $F$, we may assume without loss of generality that $T$ and $F$ have been chosen as centers.

We will now consider the following assignment $\alpha_1,\ldots,\alpha_a$: If $x_i$ is assigned to $T$, we set $\alpha_i = 1$. Otherwise we set $\alpha_i = 0$. Now consider an arbitrary clause $c_i$. Note that for any $j \leq m$ it holds that $c_{i,j}$ has distance $2$ to $F$. Thus for every $i$, there exists at least one $j$ such that $c_{i,j}$ is assigned to $T$. Since the clusters are connected there exist two cases:

\begin{itemize}
\item There exists an $i'$ such that $x_{i'}$ and $c_{i,j}$ are neighbored and $x_{i'}$ has been assigned to $T$. Then $x_{i'}$ is contained in $c_i$ and $\alpha_{i'} = 1$. Thus $c_i$ is satisfied.
\item There exists an $i'$ such that $\overline{x}_{i'}$ and $c_{i,j}$ are neighbored and $\overline{x}_{i'}$ has been assigned to $T$. With a similar argument as for $c_{i,j}$, we may argue that there exists at least one $e_{i',j'}$ that has been assigned to $F$. Since $e_{i',j'}$ is only neighbored to $x_{i'}$ and $\overline{x}_{i'}$, it follows that $x_{i'}$ has been assigned to $F$ which means that $\alpha_{i'} = 0$. By combing this with the fact that $\overline{x}_{i'}$ is contained in $c_i$, we obtain that $c_i$ is satisfied.
\end{itemize}

Since this holds for any clause $c_i$ it follows that $\alpha$ satisfies every single clause and the lemma holds.
\end{proof}

Using these two lemmas we can prove the theorem.

\begin{proof}[Proof of Theorem~\ref{thm:Hardnessk2}]
Let us consider an arbitrary $\epsilon \geq 0$. By combining Lemma \ref{lem:dir1} and \ref{lem:dir2} we know that any algorithm obtaining an approximation ratio smaller $\frac{m}{a}$ would directly solve the 3-SAT problem. We now choose $m = \left(2a^2 + ab\right)^{1/\epsilon}$. Then it holds that the total number of nodes $n = 2+(m + 2) a + m b \leq (2a + b)m$. Now it holds that:
\begin{align*}
n ^{1- \epsilon} &\leq \left( (2a + b)  \left(2a^2 + ab\right)^{1/\epsilon} \right)^{1 - \epsilon}\\
&< (2a + b) \left(2a^2 + ab\right)^{(1/\epsilon) - 1}\\
&= (2a + b) m / \left(2a^2 + ab\right)\\
&= \frac{m}{a}.
\end{align*}
Thus an $O(n^{1 - \epsilon})$-approximation for the connected $k$-median problem would solve the 3-SAT problem in polynomial time.

Furthermore it holds that for $k = 2$ a polynomial-time approximation algorithm for the assignment version of the connected $k$-median problem (i.e., the version where the centers are given) would directly yield a polynomial-time algorithm for the original problem with the same approximation ratio, since one could simply try out every combination of centers in polynomial time. 
\end{proof}

At the same time the $O(\log^2 k)$-approximation algorithm for connected $k$-center by Drexler et al.~\cite{Drexler2024} directly provides an $O(n\log^2 k)$-approximation algorithm for the disjoint connected $k$-median problem. This is due to the fact the cost of the optimum $k$-median solution is lower bounded by the maximum distance of any point to its center because this distance is also a summand in the objective function. Thus the cost of the best $k$-median solution is lower bounded by the cost of the best $k$-center solution. At the same time, if we interpret a disjoint connected $k$-center solution with radius $r$ as a $k$-median solution, every node is at most adding a cost of $r$ to the objective since it is only contained in a single cluster which bounds the cost of the $k$-median objective by $O(nr)$. By combining these two bounds and using the disjoint connected $k$-center algorithm for disjoint connected $k$-median, we directly get the following result.

\begin{theorem}\label{thm:disjoint_general_apx}
    For the disjoint connected $k$-median problem, there is a polynomial-time algorithm with approximation factor $O(n\log^2{k})$.
\end{theorem}

One might note that up to some subpolynomial factors this approximation guarantee is tight.

%% file: k_median_tree.tex
In this section, we will address the connected k-median problem when the connectivity graph $G = (V,E)$ is a tree. 
We will designate an arbitrary node $r \in V$ as the root of our tree. For two adjacent nodes $a$ and $b$, 
we say that $a$ is the parent of $b$ if the unique path from $b$ to the root $r$ passes through $a$. 
All nodes $b$ with the parent node $a$ are referred to as the children of $a$ and are included in the set $c(a)$. 
A node without a child is termed a leaf. Since a tree does not contain any cycles, a subtree rooted at $a$ can be separated 
from the entire tree by removing the edge from $a$ to its parent. This resulting tree is denoted as $T_a$ and its height is 
defined as the length of the longest path between a leaf in $T_a$ and $a$. 
To simplify, we might sometimes refer to the height of $T_a$ as simply the height of the node $a$.
 
We will compute the optimal $k$-median clustering with a dynamic program, 
related to the one used for connected $k$-center clustering on trees~\cite{Drexler2024}. The latter one is based on the following insight:
For each subtree $T_a$, the node $a$ is the only node with an edge extending outside the subtree. 
Thus, due to the connectivity constraint, every node within $T_a$ can only be assigned to a center that resides in $T_a$ itself, 
or to the same center as $a$. This effectively allows to solve the $k$-center clustering problem on $T_a$ independently 
from the rest of the tree, once the center to which $a$ will be assigned is determined.

More specifically, to solve the connected $k$-center problem on a tree, one guesses the optimal radius. 
Afterwards, the number of centers required within the subtree $T_a$ is calculated for each possible assignment of $a$, 
using a bottom-up approach. To expedite calculations for $T_a$, the respective values for the children of $a$ are utilized. 
Ultimately, it can be checked whether the entire tree can be covered by $k$ centers, 
thereby determining if the selected radius was feasible.
 
For the $k$-median problem, we cannot simply guess the maximum radius of the clusters beforehand. 
Instead, we'll need to compute the cost generated by the subtree $T_a$ based on the number of available clusters. 
Specifically, we calculate the following entries:

\begin{itemize}
    \item $I(a, b, k')$ represents the minimum cost of the subtree $T_a$ 
    if its root is assigned to node $b \in T_a$, and we can select up to $k'$ centers in $T_a$ (including $b$).
    \item $I(a,k')$ represents the minimum cost of the subtree $T_a$ 
    if we can select $k'$ centers within the subtree $T_a$ and we need to assign all nodes in $T_a$ to a one of the $k'$ centers. 
    Essentially, we consider $T_a$ separate from the rest of the tree. 
    It is straightforward to see that $I(a,k') = \min_{b \in T_a} I(a,b,k')$ always holds.
    \item $C(a,b,k')$ represents the minimum cost of the subtree $T_a$  
    if we can select up to $k'$ centers within $T_a$ and we're allowed to assign nodes in $T_a$ to a center $b$ outside of $T_a$. 
    Given that any path from a node $v$ within $T_a$ to a node $b$ outside of $T_a$ must 
    pass through the parent of $a$, we can use the connectivity of the clusters to deduce that $v$ either needs to be assigned to $b$ as well,
    or $v$ is assigned to another node within $T_a$.
    \end{itemize}

Interestingly Angelidakis et al. already developed a very similar recursion as well as a dynamic program for it in the different context of stable clustering \cite{angelidakis2017algorithms}. The main technical difference between their result and ours is that they can assume that the trees are binary while to the authors' knowledge there is no easy way to transform an arbitrary tree into a binary tree in the context of connected clustering without enabling additional incorrect solutions. Thus our dynamic program needs to able to deal with nodes of arbitrary degree.

For nodes $a$ of height $0$ (i.e., the leaves), we proceed as follows:
\begin{itemize}
\item For $I(a, 0)$ and $I(a,a,0)$, the value is $\infty$ since we must choose $a$ itself as a center 
if we aim to assign it within $T_a$. However, for $k' \geq 1$, $I(a,k')$ and $I(a,a,k')$ equals $0$ 
as the distance from $a$ to itself is $0$.
\item For $C(a,b,k')$, we must distinguish whether or not $k' = 0$. If it is, we need to assign $a$ to $b\neq a$ 
and incur a cost of $d(a,b)$. Otherwise, we can choose $a$ itself as a center, which results in a cost of $C(a,b,k')=C(a,a,k')=0$.
\end{itemize}

When calculating the entries of $C$ and $I$ for a node $a$ that is not a leaf, 
we encounter an issue that was absent in the $k$-center variant. Since $a$ could have multiple children, 
we essentially need to determine how many centers should be assigned to the respective subtrees. 
However, there are too many potential combinations to explore.  
To circumvent this, we will introduce another recursion to calculate $C$ and $I$ efficiently. 
To do this, we order the children of $a$ in an arbitrary order $a_1,..,a_z$. 
Let $T(a, z')$ be a tree consisting of $a$ and the subtrees $T_{a_1},...,T_{a_{z'}}$ for node $a$ and $1\le z'\le c(a)$. 
We then define the tables $X_a$ and $Y_a$ as follows:
\begin{itemize}
\item $X_a(b, k', z')$ denotes the minimum cost for the tree $T(a, z')$ if $a$ is assigned to a center $b$ within $T(a, z')$, 
and we can choose $k'$ centers within $T(a, z')$. Similar to the definition of $I$, 
we define $X_a(k',z') = \min_{b \in T(a, z')} X_a(b, k', z')$, which is the minimum cost produced by $T(a, z')$ 
if we assign all nodes to a center within $T(a, z')$ itself, and are allowed to choose $k'$ centers.
\item $Y_a(b, k', z')$ denotes the minimum cost of the same tree $T(a, z')$ as in the definition of $X_a(b, k', z')$ 
if we are allowed to choose $k'$ centers within it, and $a$ is assigned to the fixed center $b$ outside of $T(a, z')$. 
Note that since every path from a node within $T(a, z')$ to a center outside of it must pass through $a$, 
it does actually hold that $Y_a(b, k', z')$ is independent of the choice of other centers outside of $T(a, z')$ 
and is thus well-defined.
\end{itemize}

For the initialization of the recursion, we will only consider the first child, i.e., $z'=1$:
\begin{itemize}
\item We set $X_a(b, k', 1) = d(b,a) + I(a_1,b, k')$ if $b \neq a$, 
because under this condition, $b$ must be contained in $T_{a_1}$, 
and thus the path from $b$ to $a$ must pass through $a_1$. 
Hence, by the connectivity condition, $a_1$ needs to be assigned to $b$, 
and the minimum cost for the subtree $T_{a_1}$ equals $I(a_1,b,k')$. 
Additionally, we also need to add the cost of the node $a$, 
which is obviously $d(a,b)$. If $a=b$, however, it holds that $X_a(a,k',1) = C(a_1,a,k' - 1)$, 
because under these circumstances, the parent of $a_1$ (i.e., $a$) has been assigned to $a$, 
and we are only allowed to choose $k' - 1$ centers in $T_{a_1}$ because $a$ was already chosen as a center.
\item Similarly, we set $Y_a(b,k',1) = d(b,a) + C(a_1,b, k')$. 
By the definition of $Y_a$, $a$ gets assigned to $b$, hence the cost of $a$ equals $d(a,b)$. 
For the subtree $T_{a_1}$, the parent of $A_1$ has been assigned to $b$, and we are allowed to choose $k'$ within $T_{a_1}$. 
By the definition of $C$, the minimum cost of the subtree equals $C(a_1,b,k')$ in this situation.
\end{itemize}

Now, we need to calculate $X_a(b, k', i+ 1)$ if we have already computed the corresponding entries of $X_a$ and $Y_a$ for index $i$. 
Recall that $T(a, i)$ is the tree consisting of $a$ and the subtrees $T_{a_1},...,T_{a_i}$, and $T(a, i + 1)$ is the tree consisting of $a$ and $T_{a_1},..,T_{a_{i+1}}$. 
We know that any of the $k'$ centers in $T(a, i + 1)$ is either contained in $T(a, i)$ or $T_{a_{i+1}}$, 
and for both subtrees, we have already determined the minimum cost of the contained nodes for any given number of centers (fewer than $k'$) chosen inside them. 
We note that these costs (as well as the choices of the centers in the respective subtrees) are independent of each other once we have fixed the center $b$ 
to which the root $a$ is assigned. Thus, we can try all possibilities to split up the centers among the two subtrees and take the one with the minimum cost. 
We will consider two different cases:
\begin{itemize}
    \item  If $b$ is contained in $T(a, i)$, at least one center must be in the subtree and its cost for $k_1$ available centers is equal to $X_a(b, k_1, i)$, 
    while the cost of $T_{a_{i+1}}$ with the remaining $k' -k_1$ centers is equal to $C(a_{i+1}, b, k' - k_1)$. Thus we have:
    \begin{equation*}
        X_a(b, k', i+ 1) = \min_{1 \leq k_1 \leq k'} \left[X_a(b, k_1, i) + C(a_{i+1}, b, k - k_1)\right]
        \end{equation*}
        \item If the center $b$ is in $T_{a_{i+1}}$, the cost of $T(a, i)$ with $k'$ centers equals $Y_a(b, k_1, i)$ 
        (because now $a$ is assigned to a center $b$ outside of $T(a, i)$), while the additional cost of $T_{a_{i+1}}$ equals $I(a_{i+1}, b, k' - k_1)$. 
        Furthermore, $T_{a_{i+1}}$ needs to contain at least one center. As a result, we get the following formula: 
\begin{equation*}
    X_a(b, k', i+ 1) = \min_{0 \leq k_1 \leq k' - 1} \left[ Y_a(b, k_1, i) + I(a_{i+1}, b, k - k_1) \right]
    \end{equation*}
\end{itemize}
 
To calculate $Y_a(b, k', i+ 1)$, we don't need to make this distinction since $a$ is not in $T(a, i + 1)$, 
and thus the cost of $T(a,i)$ with $k_1$ centers equals $Y_a(b, k_1, i)$ and the cost of $T_{a_{i+1}}$ equals $C(a_{i+1}, b, k' - k_1)$. 
There are no additional constraints for the choice of the centers besides the fact that their total number is exactly $k'$:
\begin{equation*}
Y_a(b, k', i+ 1) = \min_{0 \leq k_1 \leq k'} \left[ Y_a(b, k_1, i) + C(a_{i+1}, b, k' - k_1) \right]
\end{equation*}

It's worth noting that for a node $b$ within $T_a$,  $I(a,b,k') = X_a(b,k',z)$ holds. 
Similarly, $C(a,b,k')$ can be calculated as the minimum of $Y_a(b,k',z)$ (i.e., the cost if $a$ is also assigned to $b$) 
and $I(a,k') = \min_{b' \in T_a} I(a,b',k')$ (the minimum cost if $a$ is assigned to a node within $T_a$).

To calculate the cost of the optimum clustering, we use dynamic programming: The algorithm considers all nodes sorted by their height in ascending order. 
The results are stored in a suitable table. If the respective node $a$ has height $0$ and thus no children, 
we can calculate for every $b \in V \setminus a$ and $k' \leq k$ the value of $C(a,b,k')$ as well as $I(a,a,k')$ directly using the initialization of the recursion.

Otherwise, let $z$ be the number of children of $a$. Using another dynamic program, 
we can calculate for increasing $z' \leq z$, $k' \leq k$ and $b \in V$ the value of $X_a(b,k',z')$ if $b \in T(a, z')$ and $Y_a(b,k',z)$ otherwise. 
The respective calculations use the recursive formula as well as the previously stored values. Once the values of $X_a$ and $Y_a$ are calculated for $z' = z$, 
we can use them to easily obtain the values of $C(a,b,k')$ for any $k' \leq k$, $b \not\in T_a$ and $I(a,b,k')$ for any $k' \leq k$, $b \in T_a$.

At the end, the optimum cost of the $k$-median clustering is equal to $I(r,k) = \min_{v \in V} I(r,v,k)$ where $r$ denotes the root of the entire tree.

Regarding the running time, we first consider the costs produced by a single node $a$. 
For each $z' \leq z$, $k' \leq k$, and $b \in V$, we need to calculate either $X_a(b,k',z')$ or $Y_a(b,k',z')$. 
Using the respective entries of the table for the smaller $z'$, we can calculate the respective value in $O(k')$ since we have to minimize over $k'$ different values $k_1 \leq k'$ 
and only need to do constant calculations for each of those. 
In total, the entire running time for node $a$ is upper bounded by $O(nk^2(c(a) + 1))$ 
where $c(a)$ is the number of children of $a$. Thus the entire cost can be bounded by the following:

\begin{align*}
\sum_{a \in V} O(nk^2(c(a) + 1)) = nk^2 \left( \sum_{a \in V} O(1) + \sum_{a \in V} O(c(a))\right)= nk^2 O(n) = O(n^2k^2)
\end{align*}

We used here that the sum of all $c(a)$ over all nodes $a$ is equal to $n-1$ 
because every node is a child of exactly one other node (except for the root). 
In conclusion, the algorithm is able to solve the problem in $O(n^2k^2)$ time, proving the following theorem. 

\DPdisjoint*

One might note that during the entire proof we never used that the distances form a metric. In fact the distances can be chosen arbitrary and the proof still works. Also if the centers are given, the problem gets even easier. Instead of having to calculate the costs produced by a subtree for every possible number of centers in it, we now obtain a single cost value for a fixed assignment of the root of the subtree. Because of this, the inner dynamic program can even be skipped entirely and the resulting dynamic program for the assignment version mostly resembles the one for the connected $k$-center objective. The running time decreases to $O(nk)$, the number of subtrees multiplied with the number of center choices.

%% file: alt_lp_flows.tex
\section{Polynomial Linear Program for the Non-Disjoint Connected $k$-Median Assignment Problem}\label{sec:NDCkMAP-alternative-LP-formulation}
In Section \ref{sec:NDCkMAP}, we used an exponentially large separator-based LP formulation for the assignment problem of the non-disjoint connected $k$-median problem. This is used because it is easier to prove our results. Here, we give an alternative flow-based LP formulation that has polynomial size and allows us to use a broader spectrum of algorithms that solve linear programs:

\begin{alignat*}{3}
    \min \quad && \multicolumn{2}{c}{ $\displaystyle\sum\limits_{v \in V,c \in C} d(v,c) x_v^c$ } \\
    \textrm{s.t.} \quad &&x_c^c & =1, && \quad\forall c \in C, \\
    &&x_v^v & =0, && \quad\forall v \in V \backslash C, \\
    &&\sum\limits_{c \in C} x_v^c & \geq 1, && \quad\forall v \in V,\\
    &&\sum\limits_{w \in N(v)} y_{v,w}^{v,c} & =x_v^c, && \quad\forall c \in C, v \in V,\\
    &&\sum\limits_{w \in N(u)} y_{w,u}^{v,c} &= \sum\limits_{w \in N(u)} y_{u,w}^{v,c}, && \quad\forall c \in C, v \in V, u \in V \backslash \{v,c\}, \\
    &&\sum\limits_{w \in N(u)} y_{w,u}^{v,c} & \leq x_u^c, && \quad\forall c \in C, v \in V, u \in V \backslash \{v,c\},\\
    &&x_v^c &\in \{0,1\}, && \quad\forall c \in C, v \in V, \\
    &&y_{u,w}^{v,c} & \in \{0,1\}, && \quad\forall v \in V, c \in C, \{u,w\} \in E,
\end{alignat*}
where $N(v) := \{u \in V \mid \exists \{u,v\} \in E\}$.
The variables $x_v^c$ encode, whether point $v$ is assigned to the center $c$.
The flow constraints forces us that there is a flow of value $x_v^c$ from $v$ to $c$.
This $v$-$c$-flow is encoded by the variables $y^{v,c}_{u,w}$ for $\{u,w\} \in E$.
Here, $C$ defines the set of $k$ centers and thus, we fix $x_c^c = 1$ for all $c \in C$.
Using the max-flow min-cut theorem \cite{ford1956maximal} it is easy to verify that there exists a flows from $v$ to $c$ of value $x_v^c$ if and only if the value of the minimum $(v,c)$-cut is at least $x_v^c$. Thus the values $\left(x_v^c\right)_{v \in V, c \in C}$ of this alternative formulation also form a solution for LP \eqref{LP_nd_assignment}.